\documentclass[11pt, a4paper,onecolumn]{quantumarticle}
\pdfoutput=1
\usepackage{revsymb, amsmath, amsfonts, amssymb, enumerate, fullpage, amsthm, graphicx, braket, relsize, bbm, mathrsfs, mathtools}
\usepackage{graphicx,color}
\definecolor{darkbrown}{rgb}{0.4, 0.26, 0.13}
\definecolor{ao}{rgb}{0.0, 0.5, 0.0}
\definecolor{bleudefrance}{rgb}{0.19, 0.55, 0.91}

\usepackage{paralist}

\usepackage[linktocpage=true, colorlinks=true, linkcolor=blue, urlcolor=blue, citecolor=blue]{hyperref}

\newtheorem{theo}{Theorem}
\newtheorem{thm}[theo]{Theorem}
\newtheorem{prop}[theo]{Proposition} 
\newtheorem{lem}[theo]{Lemma}
\newtheorem{cor}[theo]{Corollary}

\newtheorem{defn}[theo]{Definition}

\newcommand{\cH}{\mathcal{H}}

\newcommand{\id}{\mathbb{I}}
\newcommand{\tr}[2]{\mathrm{tr}_{#2} \left\{ #1 \right\}}
\newcommand{\Tr}[1]{\mathrm{tr}\left\{#1 \right\}}

\newcommand{\cS}{\mathcal{S}}
\newcommand{\cE}{\mathcal{E}}
\newcommand{\cQ}{\mathcal{Q}}
\newcommand{\X}{\mathbb{X}}
\newcommand{\Y}{\mathbb{Y}}
\newcommand{\A}{\mathbb{A}}
\newcommand{\Bo}{\mathbb{B}}
\newcommand{\As}{\boldsymbol{\Sigma}}
\newcommand{\Op}{\mathbb{O}}
\newcommand{\Nb}{\mathrm{N}}
\newcommand{\Pro}{\mathrm{\textbf{P}}}

\newcommand{\blk}{\color{black}}

\definecolor{orangy}{RGB}{213,94,0}

\usepackage{cite} 

\usepackage{rotating}
\usepackage{floatpag}
\usepackage{comment}
\rotfloatpagestyle{empty}
\renewcommand{\theenumi}{\arabic{enumi}}
\usepackage{hyperref}

\usepackage{amsmath,amsthm,amssymb}
\usepackage{xspace,enumerate,color,epsfig}
\usepackage{graphicx}
\graphicspath{{.}{./figures/}}
\usepackage{marvosym}

\usepackage{tikzfig}
\usepackage{stmaryrd}
\usepackage{docmute}
\usepackage{keycommand}

\usepackage{pifont}

\usepackage{enumitem}

\newcommand{\C}{\ensuremath{\mathbb{C}}}

\newkeycommand{\pointmap}[style=point][1]{\,\tikz{\node[style=\commandkey{style}] (x) at (0,-0.3) {$#1$};\node [style=none] (3) at (0, 1.0) {};\node [style=none] (3) at (0, -1.0) {};\draw (x)--(0,0.7);}\,}

\newkeycommand{\typedkpoint}[style=kpoint,edgestyle=][2]{%
\,\begin{tikzpicture}
  \begin{pgfonlayer}{nodelayer}
    \node [style=none] (0) at (0, 1) {};
    \node [style=\commandkey{style}] (1) at (0, -0.75) {$#2$};
    \node [style=right label] (2) at (0.25, 0.5) {$#1$};
  \end{pgfonlayer}
  \begin{pgfonlayer}{edgelayer}
    \draw [\commandkey{edgestyle}] (1) to (0.center);
  \end{pgfonlayer}
\end{tikzpicture}\,}

\newkeycommand{\typedkpointdag}[style=kpoint adjoint,edgestyle=][2]{%
\,\begin{tikzpicture}
  \begin{pgfonlayer}{nodelayer}
    \node [style=none] (0) at (0, -1) {};
    \node [style=\commandkey{style}] (1) at (0, 0.75) {$#2$};
    \node [style=right label] (2) at (0.25, -0.5) {$#1$};
  \end{pgfonlayer}
  \begin{pgfonlayer}{edgelayer}
    \draw [\commandkey{edgestyle}] (0.center) to (1);
  \end{pgfonlayer}
\end{tikzpicture}\,}

\newcommand{\bistateadj}[1]{\,\begin{tikzpicture}[yshift=-3mm]
  \begin{pgfonlayer}{nodelayer}
    \node [style=kpointadj, minimum width=1 cm, inner sep=2pt] (0) at (0, 0.5) {$\psi$};
    \node [style=none] (1) at (-0.75, 0.25) {};
    \node [style=none] (2) at (0.75, 0.25) {};
    \node [style=none] (3) at (-0.75, -0.5) {};
    \node [style=none] (4) at (0.75, -0.5) {};
  \end{pgfonlayer}
  \begin{pgfonlayer}{edgelayer}
    \draw (1.center) to (3.center);
    \draw (2.center) to (4.center);
  \end{pgfonlayer}
\end{tikzpicture}\,}

\newkeycommand{\pointketbra}[style1=copoint,style2=point][2]{\,%
\begin{tikzpicture}
  \begin{pgfonlayer}{nodelayer}
    \node [style=none] (0) at (0, -1.5) {};
    \node [style=\commandkey{style1}] (1) at (0, -0.75) {$#1$};
    \node [style=\commandkey{style2}] (2) at (0, 0.75) {$#2$};
    \node [style=none] (3) at (0, 1.5) {};
  \end{pgfonlayer}
  \begin{pgfonlayer}{edgelayer}
    \draw (0.center) to (1);
    \draw (2) to (3.center);
  \end{pgfonlayer}
\end{tikzpicture}\,}

\newkeycommand{\twocopointketbra}[style1=copoint,style2=copoint,style3=point][3]{\,%
\begin{tikzpicture}
  \begin{pgfonlayer}{nodelayer}
    \node [style=none] (0) at (-0.75, -1.5) {};
    \node [style=none] (1) at (0.75, -1.5) {};
    \node [style=\commandkey{style1}] (2) at (-0.75, -0.75) {$#1$};
    \node [style=\commandkey{style2}] (3) at (0.75, -0.75) {$#2$};
    \node [style=\commandkey{style3}] (4) at (0, 0.75) {$#3$};
    \node [style=none] (5) at (0, 1.5) {};
  \end{pgfonlayer}
  \begin{pgfonlayer}{edgelayer}
    \draw (0.center) to (2);
    \draw (1.center) to (3);
    \draw (4) to (5.center);
  \end{pgfonlayer}
\end{tikzpicture}\,}

\newkeycommand{\twopointketbra}[style1=copoint,style2=point,style3=point][3]{\,%
\begin{tikzpicture}
  \begin{pgfonlayer}{nodelayer}
    \node [style=none] (0) at (0, -1.5) {};
    \node [style=none] (1) at (-0.75, 1.5) {};
    \node [style=\commandkey{style1}] (2) at (0, -0.75) {$#1$};
    \node [style=\commandkey{style2}] (3) at (-0.75, 0.75) {$#2$};
    \node [style=\commandkey{style3}] (4) at (0.75, 0.75) {$#3$};
    \node [style=none] (5) at (0.75, 1.5) {};
  \end{pgfonlayer}
  \begin{pgfonlayer}{edgelayer}
    \draw (0.center) to (2);
    \draw (1.center) to (3);
    \draw (4) to (5.center);
  \end{pgfonlayer}
\end{tikzpicture}\,}

\newkeycommand{\pointbraket}[style1=point,style2=copoint][2]{\,%
\begin{tikzpicture}
  \begin{pgfonlayer}{nodelayer}
    \node [style=\commandkey{style1}] (0) at (0, -0.5) {$#1$};
    \node [style=\commandkey{style2}] (1) at (0, 0.5) {$#2$};
  \end{pgfonlayer}
  \begin{pgfonlayer}{edgelayer}
    \draw (0) to (1);
  \end{pgfonlayer}
\end{tikzpicture}\,}

\newkeycommand{\kpointbraket}[style1=kpoint,style2=kpoint adjoint][2]{\,%
\begin{tikzpicture}
  \begin{pgfonlayer}{nodelayer}
    \node [style=\commandkey{style1}] (0) at (0, -0.75) {$#1$};
    \node [style=\commandkey{style2}] (1) at (0, 0.75) {$#2$};
  \end{pgfonlayer}
  \begin{pgfonlayer}{edgelayer}
    \draw (0) to (1);
  \end{pgfonlayer}
\end{tikzpicture}\,}

\newkeycommand{\kpointmap}[style1=kpoint,style2=map][2]{\,%
\begin{tikzpicture}
  \begin{pgfonlayer}{nodelayer}
    \node [style=\commandkey{style1}] (0) at (0, -1.1) {$#1$};
    \node [style=\commandkey{style2}] (1) at (0, 0.2) {$#2$};
    \node [style=none] (2) at (0, 1.5) {};
  \end{pgfonlayer}
  \begin{pgfonlayer}{edgelayer}
    \draw (0) to (1);
    \draw (1) to (2.center);
  \end{pgfonlayer}
\end{tikzpicture}%
\,}

\newkeycommand{\sandwichtwo}[style1=point,style2=map,style3=map,style4=copoint][4]{\,%
\begin{tikzpicture}
  \begin{pgfonlayer}{nodelayer}
    \node [style=\commandkey{style2}] (0) at (0, -0.75) {$#2$};
    \node [style=\commandkey{style3}] (1) at (0, 0.75) {$#3$};
    \node [style=\commandkey{style1}] (2) at (0, -2) {$#1$};
    \node [style=\commandkey{style4}] (3) at (0, 2) {$#4$};
  \end{pgfonlayer}
  \begin{pgfonlayer}{edgelayer}
    \draw (2) to (0);
    \draw (0) to (1);
    \draw (1) to (3);
  \end{pgfonlayer}
\end{tikzpicture}\,}

\newkeycommand{\longbraket}[style1=point,style2=copoint][2]{\,%
\begin{tikzpicture}
  \begin{pgfonlayer}{nodelayer}
    \node [style=\commandkey{style1}] (0) at (0, -2) {$#1$};
    \node [style=\commandkey{style2}] (1) at (0, 2) {$#2$};
  \end{pgfonlayer}
  \begin{pgfonlayer}{edgelayer}
    \draw (0) to (1);
  \end{pgfonlayer}
\end{tikzpicture}\,}

\newkeycommand{\onb}[style=point]{\ensuremath{\left\{\tikz{\node[style=\commandkey{style}] (x) at (0,-0.3) {$j$};\draw (x)--(0,0.7);}\right\}}\xspace}

\newkeycommand{\copointmap}[style=copoint][1]{\,\tikz{\node[style=\commandkey{style}] (x) at (0,0.3) {$#1$};\draw (x)--(0,-0.7);}\,}

\tikzstyle{cdiag}=[matrix of math nodes, row sep=3em, column sep=3em, text height=1.5ex, text depth=0.25ex,inner sep=0.5em]
\tikzstyle{arrow above}=[transform canvas={yshift=0.5ex}]
\tikzstyle{arrow below}=[transform canvas={yshift=-0.5ex}]

\usetikzlibrary{shapes.multipart}

\tikzset{->-/.style={decoration={
  markings,
  mark=at position .5 with {\arrow{>}}},postaction={decorate}}}
\tikzset{-<-/.style={decoration={
  markings,
  mark=at position .5 with {\arrow{<}}},postaction={decorate}}}

\tikzstyle{bwSpider}=[
       rectangle split,
       rectangle split parts=2,
       rectangle split part fill={black,white},
 minimum size=3.6 mm, inner sep=-2mm, draw=black,scale=0.5,rounded corners=0.8 mm
       ]
 \tikzstyle{wbSpider}=[
       rectangle split,
       rectangle split parts=2,
       rectangle split part fill={white,black},
 minimum size=3.6 mm, inner sep=-2mm, draw=black,scale=0.5,rounded corners=0.8 mm
       ]

\tikzstyle{epiCopoint}=[regular polygon,regular polygon sides=3,draw,scale=0.75,inner sep=-0.5pt,minimum width=5mm,fill=white,regular polygon rotate=0,line width=1pt]
\tikzstyle{epiPoint}=[regular polygon,regular polygon sides=3,draw,scale=0.75,inner sep=-0.5pt,minimum width=5mm,fill=white,regular polygon rotate=180,line width=1pt]
\tikzstyle{epiPointWide}=[regular polygon,regular polygon sides=3,draw,scale=0.75,inner sep=-0.5pt,minimum width=8mm,fill=white,regular polygon rotate=180,line width=1pt]
\tikzstyle{epiBox}=[fill=white,draw, line width = 1pt,inner sep=0.6mm,font=\footnotesize,minimum height=3mm,minimum width=3mm]
\tikzstyle{epiBoxWide}=[fill=white,draw, line width = 1pt,inner sep=0.6mm,font=\footnotesize,minimum height=3mm,minimum width=5mm]
\tikzstyle{epiBoxVeryWide}=[fill=white,draw, line width = 1pt,inner sep=0.6mm,font=\footnotesize,minimum height=3mm,minimum width=7mm]
\tikzstyle{qWire}=[line width = 1pt, color=black]
\tikzstyle{cWire}=[color=gray,line width = .75pt]
\tikzstyle{CqWire}=[color=gray,line width = .75pt,->-]
\tikzstyle{CcWire}=[color=gray,line width = .75pt,->-]
\tikzstyle{RqWire}=[line width = 1pt, color=black,-<-]
\tikzstyle{RcWire}=[color=gray,line width = .75pt,-<-]
\tikzstyle{env}=[copoint,regular polygon rotate=0,minimum width=0.2cm, fill=black]

\tikzstyle{probs}=[shape=semicircle,fill=white,draw=black,shape border rotate=180,minimum width=1.2cm]

\tikzstyle{every picture}=[baseline=-0.25em,scale=0.5]
\tikzstyle{dotpic}=[] 
\tikzstyle{diredges}=[every to/.style={diredge}]
\tikzstyle{math matrix}=[matrix of math nodes,left delimiter=(,right delimiter=),inner sep=2pt,column sep=1em,row sep=0.5em,nodes={inner sep=0pt},text height=1.5ex, text depth=0.25ex]

\tikzstyle{inline text}=[text height=1.5ex, text depth=0.25ex,yshift=0.5mm]
\tikzstyle{label}=[font=\footnotesize,text height=1.5ex, text depth=0.25ex,yshift=0.5mm]
\tikzstyle{left label}=[label,anchor=east,xshift=1.5mm]
\tikzstyle{right label}=[label,anchor=west,xshift=-1mm]
\tikzstyle{up label}=[label,anchor=south,yshift=-1mm]

\tikzstyle{braceedge}=[decorate,decoration={brace,amplitude=2mm,raise=-1mm}]
\tikzstyle{small braceedge}=[decorate,decoration={brace,amplitude=1mm,raise=-1mm}]

\tikzstyle{doubled}=[line width=1.6pt] 
\tikzstyle{boldedge}=[doubled,shorten <=-0.17mm,shorten >=-0.17mm]
\tikzstyle{boldedgegray}=[doubled,gray,shorten <=-0.17mm,shorten >=-0.17mm]
\tikzstyle{singleedgegray}=[gray]

\tikzstyle{semidoubled}=[line width=1.4pt] 
\tikzstyle{semiboldedgegray}=[semidoubled,gray,shorten <=-0.17mm,shorten >=-0.17mm]

\tikzstyle{boxedge}=[semiboldedgegray]

\tikzstyle{boldedgedashed}=[very thick,dashed,shorten <=-0.17mm,shorten >=-0.17mm]
\tikzstyle{vboldedgedashed}=[doubled,dashed,shorten <=-0.17mm,shorten >=-0.17mm]
\tikzstyle{left hook arrow}=[left hook-latex]
\tikzstyle{right hook arrow}=[right hook-latex]
\tikzstyle{sembracket}=[line width=0.5pt,shorten <=-0.07mm,shorten >=-0.07mm]

\tikzstyle{causal edge}=[->,thick,gray]
\tikzstyle{causal nondir}=[thick,gray]
\tikzstyle{timeline}=[thick,gray, dashed]

\tikzstyle{cedge}=[<->,thick,gray!70!white]

\tikzstyle{empty diagram}=[draw=gray!40!white,dashed,shape=rectangle,minimum width=1cm,minimum height=1cm]
\tikzstyle{empty diagram small}=[draw=gray!50!white,dashed,shape=rectangle,minimum width=0.6cm,minimum height=0.5cm]

\tikzstyle{dot}=[inner sep=0mm,minimum width=2mm,minimum height=2mm,draw,shape=circle]
\tikzstyle{bigdot}=[inner sep=0mm,minimum width=5mm,minimum height=5mm,draw,shape=circle]
\tikzstyle{leak}=[white dot, shape=regular polygon, minimum size=3.3 mm, regular polygon sides=3, outer sep=-0.2mm, regular polygon rotate=270]
\tikzstyle{proj}=[regular polygon,regular polygon sides=4,draw,scale=0.75,inner sep=-0.5pt,minimum width=6mm,fill=white]
\tikzstyle{projOut}=[regular polygon,regular polygon sides=3,draw,scale=0.75,inner sep=-0.5pt,minimum width=7.5mm,fill=white,regular polygon rotate=180]
\tikzstyle{projIn}=[regular polygon,regular polygon sides=3,draw,scale=0.75,inner sep=-0.5pt,minimum width=7.5mm,fill=white]
\tikzstyle{Vleak}=[white dot, shape=regular polygon, minimum size=3.3 mm, regular polygon sides=3, outer sep=-0.2mm, regular polygon rotate=90]
\tikzstyle{dleak}=[white dot, line width=1.6pt, shape=regular polygon, minimum size=3.3 mm, regular polygon sides=3, outer sep=-0.2mm, regular polygon rotate=270]

\tikzstyle{Wsquare}=[white dot, shape=regular polygon, rounded corners=0.8 mm, minimum size=3.3 mm, regular polygon sides=3, outer sep=-0.2mm]
\tikzstyle{Wsquareadj}=[white dot, shape=regular polygon, rounded corners=0.8 mm, minimum size=3.3 mm, regular polygon sides=3, outer sep=-0.2mm, regular polygon rotate=180]
\tikzstyle{ddot}=[inner sep=0mm, doubled, minimum width=2.5mm,minimum height=2.5mm,draw,shape=circle]

\tikzstyle{clear dot}=[dot,fill=none,text depth=-0.2mm,draw=gray, line width = .75pt]
\tikzstyle{tall clear dot}=[dot,fill=none,text depth=-0.2mm,draw=gray, line width = .75pt,shape=ellipse, minimum height=5mm]
\tikzstyle{wide clear dot}=[dot,fill=none,text depth=-0.2mm,draw=gray, line width = .75pt, shape=ellipse, minimum width = 5mm]
\tikzstyle{very wide clear dot}=[dot,fill=none,text depth=-0.2mm,draw=gray, line width = .75pt, shape=ellipse, minimum width = 7mm ]

\tikzstyle{black dot}=[dot,fill=black]
\tikzstyle{white dot}=[dot,fill=white,,text depth=-0.2mm]
\tikzstyle{white Wsquare}=[Wsquare,fill=gray,,text depth=-0.2mm]
\tikzstyle{white Wsquareadj}=[Wsquareadj,fill=white,,text depth=-0.2mm]
\tikzstyle{green dot}=[white dot] 
\tikzstyle{gray dot}=[dot,fill=gray!40!white,,text depth=-0.2mm]
\tikzstyle{red dot}=[gray dot] 

\tikzstyle{black ddot}=[ddot,fill=black]
\tikzstyle{white ddot}=[ddot,fill=white]
\tikzstyle{gray ddot}=[ddot,fill=gray!40!white]

\tikzstyle{gray edge}=[gray!60!white]

\tikzstyle{small dot}=[inner sep=0.2mm,minimum width=0pt,minimum height=0pt,draw,shape=circle]

\tikzstyle{small black dot}=[small dot,fill=black]
\tikzstyle{small white dot}=[small dot,fill=white]
\tikzstyle{small gray dot}=[small dot,fill=gray,draw=gray]

\tikzstyle{causal dot}=[inner sep=0.4mm,minimum width=0pt,minimum height=0pt,draw=white,shape=circle,fill=gray!40!white]

\tikzstyle{phase dimensions}=[minimum size=5mm,font=\footnotesize,rectangle,rounded corners=2.5mm,inner sep=0.2mm,outer sep=-2mm]
\tikzstyle{dphase dimensions}=[minimum size=5mm,font=\footnotesize,rectangle,rounded corners=2.5mm,inner sep=0.2mm,outer sep=-2mm]

\tikzstyle{white phase dot}=[dot,fill=white,phase dimensions]
\tikzstyle{white phase ddot}=[ddot,fill=white,dphase dimensions]

\tikzstyle{white rect ddot}=[draw=black,fill=white,doubled,minimum size=5mm,font=\footnotesize,rectangle,rounded corners=2.5mm,inner sep=0.2mm]
\tikzstyle{gray rect ddot}=[draw=black,fill=gray!40!white,doubled,minimum size=6mm,font=\footnotesize,rectangle,rounded corners=3mm]

\tikzstyle{gray phase dot}=[dot,fill=gray!40!white,phase dimensions]
\tikzstyle{gray phase ddot}=[ddot,fill=gray!40!white,dphase dimensions]
\tikzstyle{grey phase dot}=[gray phase dot]
\tikzstyle{grey phase ddot}=[gray phase ddot]

\tikzstyle{small phase dimensions}=[minimum size=4mm,font=\tiny,rectangle,rounded corners=2mm,inner sep=0.2mm,outer sep=-2mm]
\tikzstyle{small dphase dimensions}=[minimum size=4mm,font=\tiny,rectangle,rounded corners=2mm,inner sep=0.2mm,outer sep=-2mm]

\tikzstyle{small gray phase dot}=[dot,fill=gray!40!white,small phase dimensions]
\tikzstyle{small gray phase ddot}=[ddot,fill=gray!40!white,small dphase dimensions]

\tikzstyle{small map}=[draw,shape=rectangle,minimum height=4mm,minimum width=4mm,fill=white]

\tikzstyle{cnot}=[fill=white,shape=circle,inner sep=-1.4pt]

\tikzstyle{asym hadamard}=[fill=white,draw,shape=NEbox,inner sep=0.6mm,font=\footnotesize,minimum height=4mm]
\tikzstyle{asym hadamard conj}=[fill=white,draw,shape=NWbox,inner sep=0.6mm,font=\footnotesize,minimum height=4mm]
\tikzstyle{asym hadamard dag}=[fill=white,draw,shape=SEbox,inner sep=0.6mm,font=\footnotesize,minimum height=4mm]

\tikzstyle{hadamard}=[fill=white,draw,inner sep=0.6mm,font=\footnotesize,minimum height=4mm,minimum width=4mm]
\tikzstyle{small hadamard}=[fill=white,draw,inner sep=0.6mm,minimum height=1.5mm,minimum width=1.5mm]
\tikzstyle{small hadamard rotate}=[small hadamard,rotate=45]
\tikzstyle{dhadamard}=[hadamard,doubled]
\tikzstyle{small dhadamard}=[small hadamard,doubled]
\tikzstyle{small dhadamard rotate}=[small hadamard rotate,doubled]
\tikzstyle{antipode}=[white dot,inner sep=0.3mm,font=\footnotesize]

\tikzstyle{scalar}=[diamond,draw,inner sep=0.5pt,font=\small]
\tikzstyle{dscalar}=[diamond,doubled, draw,inner sep=0.5pt,font=\small]

\tikzstyle{small box}=[rectangle,inline text,fill=white,draw,minimum height=5mm,yshift=-0.5mm,minimum width=5mm,font=\small]
\tikzstyle{small gray box}=[small box,fill=gray!30]
\tikzstyle{medium box}=[rectangle,inline text,fill=white,draw,minimum height=5mm,yshift=-0.5mm,minimum width=10mm,font=\small]
\tikzstyle{square box}=[small box] 
\tikzstyle{medium gray box}=[small box,fill=gray!30]
\tikzstyle{semilarge box}=[rectangle,inline text,fill=white,draw,minimum height=5mm,yshift=-0.5mm,minimum width=12.5mm,font=\small]
\tikzstyle{large box}=[rectangle,inline text,fill=white,draw,minimum height=5mm,yshift=-0.5mm,minimum width=15mm,font=\small]
\tikzstyle{large gray box}=[small box,fill=gray!30]

\tikzstyle{Bayes box}=[rectangle,fill=black,draw, minimum height=3mm, minimum width=3mm]

\tikzstyle{gray square point}=[small box,fill=gray!50]

\tikzstyle{dphase box white}=[dhadamard]
\tikzstyle{dphase box gray}=[dhadamard,fill=gray!50!white]
\tikzstyle{phase box white}=[hadamard]
\tikzstyle{phase box gray}=[hadamard,fill=gray!50!white]

\tikzstyle{point}=[regular polygon,regular polygon sides=3,draw,scale=0.75,inner sep=-0.5pt,minimum width=9mm,fill=white,regular polygon rotate=180]
\tikzstyle{infpoint}=[regular polygon,regular polygon sides=3,draw,scale=0.75,inner sep=-0.5pt,minimum width=9mm,fill=white,regular polygon rotate=90]
\tikzstyle{point nosep}=[regular polygon,regular polygon sides=3,draw,scale=0.75,inner sep=-2pt,minimum width=9mm,fill=white,regular polygon rotate=180]
\tikzstyle{infcopoint}=[regular polygon,regular polygon sides=3,draw,scale=0.75,inner sep=-0.5pt,minimum width=9mm,fill=white,regular polygon rotate=270]
\tikzstyle{copoint}=[regular polygon,regular polygon sides=3,draw,scale=0.75,inner sep=-0.5pt,minimum width=9mm,fill=white]
\tikzstyle{dpoint}=[point,doubled]
\tikzstyle{dcopoint}=[copoint,doubled]

\tikzstyle{pointgrow}=[shape=cornerpoint,kpoint common,scale=0.75,inner sep=3pt]
\tikzstyle{pointgrow dag}=[shape=cornercopoint,kpoint common,scale=0.75,inner sep=3pt]

\tikzstyle{wide copoint}=[fill=white,draw,shape=isosceles triangle,shape border rotate=90,isosceles triangle stretches=true,inner sep=0pt,minimum width=1.5cm,minimum height=6.12mm]
\tikzstyle{wide point}=[fill=white,draw,shape=isosceles triangle,shape border rotate=-90,isosceles triangle stretches=true,inner sep=0pt,minimum width=1.5cm,minimum height=6.12mm,yshift=-0.0mm]
\tikzstyle{wide point plus}=[fill=white,draw,shape=isosceles triangle,shape border rotate=-90,isosceles triangle stretches=true,inner sep=0pt,minimum width=1.74cm,minimum height=7mm,yshift=-0.0mm]

\tikzstyle{wide dpoint}=[fill=white,doubled,draw,shape=isosceles triangle,shape border rotate=-90,isosceles triangle stretches=true,inner sep=0pt,minimum width=1.5cm,minimum height=6.12mm,yshift=-0.0mm]

\tikzstyle{tinypoint}=[regular polygon,regular polygon sides=3,draw,scale=0.55,inner sep=-0.15pt,minimum width=6mm,fill=white,regular polygon rotate=180]

\tikzstyle{white point}=[point]
\tikzstyle{white dpoint}=[dpoint]
\tikzstyle{green point}=[white point] 
\tikzstyle{white copoint}=[copoint]
\tikzstyle{gray point}=[point,fill=gray!40!white]
\tikzstyle{gray dpoint}=[gray point,doubled]
\tikzstyle{red point}=[gray point] 
\tikzstyle{gray copoint}=[copoint,fill=gray!40!white]
\tikzstyle{gray dcopoint}=[gray copoint,doubled]

\tikzstyle{white point guide}=[regular polygon,regular polygon sides=3,font=\scriptsize,draw,scale=0.65,inner sep=-0.5pt,minimum width=9mm,fill=white,regular polygon rotate=180]

\tikzstyle{black point}=[point,fill=black,font=\color{white}]
\tikzstyle{black copoint}=[copoint,fill=black,font=\color{white}]

\tikzstyle{tiny gray point}=[tinypoint,fill=gray!40!white]

\tikzstyle{diredge}=[->]
\tikzstyle{ddiredge}=[<->]
\tikzstyle{rdiredge}=[<-]
\tikzstyle{thickdiredge}=[->, very thick]
\tikzstyle{pointer edge}=[->,very thick,gray]
\tikzstyle{pointer edge part}=[very thick,gray]
\tikzstyle{dashed edge}=[dashed]
\tikzstyle{thick dashed edge}=[very thick,dashed]
\tikzstyle{thick gray dashed edge}=[thick dashed edge,gray!40]
\tikzstyle{thick map edge}=[very thick,|->]

\makeatletter
\newcommand{\boxshape}[3]{%
\pgfdeclareshape{#1}{
\inheritsavedanchors[from=rectangle] 
\inheritanchorborder[from=rectangle]
\inheritanchor[from=rectangle]{center}
\inheritanchor[from=rectangle]{north}
\inheritanchor[from=rectangle]{south}
\inheritanchor[from=rectangle]{west}
\inheritanchor[from=rectangle]{east}
\backgroundpath{
\southwest \pgf@xa=\pgf@x \pgf@ya=\pgf@y
\northeast \pgf@xb=\pgf@x \pgf@yb=\pgf@y

\@tempdima=#2
\@tempdimb=#3

\pgfpathmoveto{\pgfpoint{\pgf@xa - 5pt + \@tempdima}{\pgf@ya}}
\pgfpathlineto{\pgfpoint{\pgf@xa - 5pt - \@tempdima}{\pgf@yb}}
\pgfpathlineto{\pgfpoint{\pgf@xb + 5pt + \@tempdimb}{\pgf@yb}}
\pgfpathlineto{\pgfpoint{\pgf@xb + 5pt - \@tempdimb}{\pgf@ya}}
\pgfpathlineto{\pgfpoint{\pgf@xa - 5pt + \@tempdima}{\pgf@ya}}
\pgfpathclose
}
}}

\boxshape{NEbox}{0pt}{3pt}
\boxshape{SEbox}{0pt}{-3pt}
\boxshape{NWbox}{5pt}{0pt}
\boxshape{SWbox}{-5pt}{0pt}
\boxshape{EBox}{-3pt}{3pt}
\boxshape{WBox}{3pt}{-3pt}
\makeatother

\tikzstyle{cloud}=[shape=cloud,draw,minimum width=1.5cm,minimum height=1.5cm]

\tikzstyle{map}=[draw,shape=NEbox,inner sep=1pt,minimum height=4mm,fill=white]
\tikzstyle{dashedmap}=[draw,dashed,shape=NEbox,inner sep=2pt,minimum height=6mm,fill=white]
\tikzstyle{mapdag}=[draw,shape=SEbox,inner sep=1pt,minimum height=4mm,fill=white]
\tikzstyle{mapadj}=[draw,shape=SEbox,inner sep=2pt,minimum height=6mm,fill=white]
\tikzstyle{maptrans}=[draw,shape=SWbox,inner sep=2pt,minimum height=6mm,fill=white]
\tikzstyle{mapconj}=[draw,shape=NWbox,inner sep=2pt,minimum height=6mm,fill=white]

\tikzstyle{medium map}=[draw,shape=NEbox,inner sep=2pt,minimum height=6mm,fill=white,minimum width=7mm]
\tikzstyle{medium map dag}=[draw,shape=SEbox,inner sep=2pt,minimum height=6mm,fill=white,minimum width=7mm]
\tikzstyle{medium map adj}=[draw,shape=SEbox,inner sep=2pt,minimum height=6mm,fill=white,minimum width=7mm]
\tikzstyle{medium map trans}=[draw,shape=SWbox,inner sep=2pt,minimum height=6mm,fill=white,minimum width=7mm]
\tikzstyle{medium map conj}=[draw,shape=NWbox,inner sep=2pt,minimum height=6mm,fill=white,minimum width=7mm]
\tikzstyle{semilarge map}=[draw,shape=NEbox,inner sep=2pt,minimum height=6mm,fill=white,minimum width=9.5mm]
\tikzstyle{semilarge map trans}=[draw,shape=SWbox,inner sep=2pt,minimum height=6mm,fill=white,minimum width=9.5mm]
\tikzstyle{semilarge map adj}=[draw,shape=SEbox,inner sep=2pt,minimum height=6mm,fill=white,minimum width=9.5mm]
\tikzstyle{semilarge map dag}=[draw,shape=SEbox,inner sep=2pt,minimum height=6mm,fill=white,minimum width=9.5mm]
\tikzstyle{semilarge map conj}=[draw,shape=NWbox,inner sep=2pt,minimum height=6mm,fill=white,minimum width=9.5mm]
\tikzstyle{large map}=[draw,shape=NEbox,inner sep=2pt,minimum height=6mm,fill=white,minimum width=12mm]
\tikzstyle{large map conj}=[draw,shape=NWbox,inner sep=2pt,minimum height=6mm,fill=white,minimum width=12mm]
\tikzstyle{very large map}=[draw,shape=NEbox,inner sep=2pt,minimum height=6mm,fill=white,minimum width=17mm]

\tikzstyle{medium dmap}=[draw,doubled,shape=NEbox,inner sep=2pt,minimum height=6mm,fill=white,minimum width=7mm]
\tikzstyle{medium dmap dag}=[draw,doubled,shape=SEbox,inner sep=2pt,minimum height=6mm,fill=white,minimum width=7mm]
\tikzstyle{medium dmap adj}=[draw,doubled,shape=SEbox,inner sep=2pt,minimum height=6mm,fill=white,minimum width=7mm]
\tikzstyle{medium dmap trans}=[draw,doubled,shape=SWbox,inner sep=2pt,minimum height=6mm,fill=white,minimum width=7mm]
\tikzstyle{medium dmap conj}=[draw,doubled,shape=NWbox,inner sep=2pt,minimum height=6mm,fill=white,minimum width=7mm]
\tikzstyle{semilarge dmap}=[draw,doubled,shape=NEbox,inner sep=2pt,minimum height=6mm,fill=white,minimum width=9.5mm]
\tikzstyle{semilarge dmap trans}=[draw,doubled,shape=SWbox,inner sep=2pt,minimum height=6mm,fill=white,minimum width=9.5mm]
\tikzstyle{semilarge dmap adj}=[draw,doubled,shape=SEbox,inner sep=2pt,minimum height=6mm,fill=white,minimum width=9.5mm]
\tikzstyle{semilarge dmap dag}=[draw,doubled,shape=SEbox,inner sep=2pt,minimum height=6mm,fill=white,minimum width=9.5mm]
\tikzstyle{semilarge dmap conj}=[draw,doubled,shape=NWbox,inner sep=2pt,minimum height=6mm,fill=white,minimum width=9.5mm]
\tikzstyle{large dmap}=[draw,doubled,shape=NEbox,inner sep=2pt,minimum height=6mm,fill=white,minimum width=12mm]
\tikzstyle{large dmap conj}=[draw,doubled,shape=NWbox,inner sep=2pt,minimum height=6mm,fill=white,minimum width=12mm]
\tikzstyle{large dmap trans}=[draw,doubled,shape=SWbox,inner sep=2pt,minimum height=6mm,fill=white,minimum width=12mm]
\tikzstyle{large dmap adj}=[draw,doubled,shape=SEbox,inner sep=2pt,minimum height=6mm,fill=white,minimum width=12mm]
\tikzstyle{large dmap dag}=[draw,doubled,shape=SEbox,inner sep=2pt,minimum height=6mm,fill=white,minimum width=12mm]
\tikzstyle{very large dmap}=[draw,doubled,shape=NEbox,inner sep=2pt,minimum height=6mm,fill=white,minimum width=19.5mm]

\tikzstyle{muxbox}=[draw,shape=rectangle,minimum height=3mm,minimum width=3mm,fill=white]
\tikzstyle{dmuxbox}=[muxbox,doubled]

\tikzstyle{box}=[draw,shape=rectangle,inner sep=2pt,minimum height=6mm,minimum width=6mm,fill=white]
\tikzstyle{dbox}=[draw,doubled,shape=rectangle,inner sep=2pt,minimum height=6mm,minimum width=6mm,fill=white]
\tikzstyle{dmap}=[draw,doubled,shape=NEbox,inner sep=2pt,minimum height=6mm,fill=white]
\tikzstyle{dmapdag}=[draw,doubled,shape=SEbox,inner sep=2pt,minimum height=6mm,fill=white]
\tikzstyle{dmapadj}=[draw,doubled,shape=SEbox,inner sep=2pt,minimum height=6mm,fill=white]
\tikzstyle{dmaptrans}=[draw,doubled,shape=SWbox,inner sep=2pt,minimum height=6mm,fill=white]
\tikzstyle{dmapconj}=[draw,doubled,shape=NWbox,inner sep=2pt,minimum height=6mm,fill=white]

\tikzstyle{ddmap}=[draw,doubled,dashed,shape=NEbox,inner sep=2pt,minimum height=6mm,fill=white]
\tikzstyle{ddmapdag}=[draw,doubled,dashed,shape=SEbox,inner sep=2pt,minimum height=6mm,fill=white]
\tikzstyle{ddmapadj}=[draw,doubled,dashed,shape=SEbox,inner sep=2pt,minimum height=6mm,fill=white]
\tikzstyle{ddmaptrans}=[draw,doubled,dashed,shape=SWbox,inner sep=2pt,minimum height=6mm,fill=white]
\tikzstyle{ddmapconj}=[draw,doubled,dashed,shape=NWbox,inner sep=2pt,minimum height=6mm,fill=white]

\boxshape{sNEbox}{0pt}{3pt}
\boxshape{sSEbox}{0pt}{-3pt}
\boxshape{sNWbox}{3pt}{0pt}
\boxshape{sSWbox}{-3pt}{0pt}
\tikzstyle{smap}=[draw,shape=sNEbox,fill=white]
\tikzstyle{smapdag}=[draw,shape=sSEbox,fill=white]
\tikzstyle{smapadj}=[draw,shape=sSEbox,fill=white]
\tikzstyle{smaptrans}=[draw,shape=sSWbox,fill=white]
\tikzstyle{smapconj}=[draw,shape=sNWbox,fill=white]

\tikzstyle{dsmap}=[draw,dashed,shape=sNEbox,fill=white]
\tikzstyle{dsmapdag}=[draw,dashed,shape=sSEbox,fill=white]
\tikzstyle{dsmaptrans}=[draw,dashed,shape=sSWbox,fill=white]
\tikzstyle{dsmapconj}=[draw,dashed,shape=sNWbox,fill=white]

\boxshape{mNEbox}{0pt}{10pt}
\boxshape{mSEbox}{0pt}{-10pt}
\boxshape{mNWbox}{10pt}{0pt}
\boxshape{mSWbox}{-10pt}{0pt}
\tikzstyle{mmap}=[draw,shape=mNEbox]
\tikzstyle{mmapdag}=[draw,shape=mSEbox]
\tikzstyle{mmaptrans}=[draw,shape=mSWbox]
\tikzstyle{mmapconj}=[draw,shape=mNWbox]

\tikzstyle{mmapgray}=[draw,fill=gray!40!white,shape=mNEbox]
\tikzstyle{smapgray}=[draw,fill=gray!40!white,shape=sNEbox]

\makeatletter

\pgfdeclareshape{cornerpoint}{
\inheritsavedanchors[from=rectangle] 
\inheritanchorborder[from=rectangle]
\inheritanchor[from=rectangle]{center}
\inheritanchor[from=rectangle]{north}
\inheritanchor[from=rectangle]{south}
\inheritanchor[from=rectangle]{west}
\inheritanchor[from=rectangle]{east}
\backgroundpath{
\southwest \pgf@xa=\pgf@x \pgf@ya=\pgf@y
\northeast \pgf@xb=\pgf@x \pgf@yb=\pgf@y

\pgfmathsetmacro{\pgf@shorten@left}{\pgfkeysvalueof{/tikz/shorten left}}
\pgfmathsetmacro{\pgf@shorten@right}{\pgfkeysvalueof{/tikz/shorten right}}

\pgfpathmoveto{\pgfpoint{0.5 * (\pgf@xa + \pgf@xb)}{\pgf@ya - 5pt}}
\pgfpathlineto{\pgfpoint{\pgf@xa - 8pt + \pgf@shorten@left}{\pgf@yb - 1.5 * \pgf@shorten@left}}
\pgfpathlineto{\pgfpoint{\pgf@xa - 8pt + \pgf@shorten@left}{\pgf@yb}}
\pgfpathlineto{\pgfpoint{\pgf@xb + 8pt - \pgf@shorten@right}{\pgf@yb}}
\pgfpathlineto{\pgfpoint{\pgf@xb + 8pt - \pgf@shorten@right}{\pgf@yb - 1.5 * \pgf@shorten@right}}
\pgfpathclose
}
}

\pgfdeclareshape{cornercopoint}{
\inheritsavedanchors[from=rectangle] 
\inheritanchorborder[from=rectangle]
\inheritanchor[from=rectangle]{center}
\inheritanchor[from=rectangle]{north}
\inheritanchor[from=rectangle]{south}
\inheritanchor[from=rectangle]{west}
\inheritanchor[from=rectangle]{east}
\backgroundpath{
\southwest \pgf@xa=\pgf@x \pgf@ya=\pgf@y
\northeast \pgf@xb=\pgf@x \pgf@yb=\pgf@y

\pgfmathsetmacro{\pgf@shorten@left}{\pgfkeysvalueof{/tikz/shorten left}}
\pgfmathsetmacro{\pgf@shorten@right}{\pgfkeysvalueof{/tikz/shorten right}}

\pgfpathmoveto{\pgfpoint{0.5 * (\pgf@xa + \pgf@xb)}{\pgf@yb + 5pt}}
\pgfpathlineto{\pgfpoint{\pgf@xa - 8pt + \pgf@shorten@left}{\pgf@ya + 1.5 * \pgf@shorten@left}}
\pgfpathlineto{\pgfpoint{\pgf@xa - 8pt + \pgf@shorten@left}{\pgf@ya}}
\pgfpathlineto{\pgfpoint{\pgf@xb + 8pt - \pgf@shorten@right}{\pgf@ya}}
\pgfpathlineto{\pgfpoint{\pgf@xb + 8pt - \pgf@shorten@right}{\pgf@ya + 1.5 * \pgf@shorten@right}}
\pgfpathclose
}
}

\makeatother

\pgfkeyssetvalue{/tikz/shorten left}{0pt}
\pgfkeyssetvalue{/tikz/shorten right}{0pt}

\tikzstyle{kpoint common}=[draw,fill=white,inner sep=1pt,minimum height=4mm]
\tikzstyle{kpoint sc}=[shape=cornerpoint,kpoint common]
\tikzstyle{kpoint adjoint sc}=[shape=cornercopoint,kpoint common]
\tikzstyle{kpoint}=[shape=cornerpoint,shorten left=5pt,kpoint common]
\tikzstyle{kpoint adjoint}=[shape=cornercopoint,shorten left=5pt,kpoint common]
\tikzstyle{kpoint conjugate}=[shape=cornerpoint,shorten right=5pt,kpoint common]
\tikzstyle{kpoint transpose}=[shape=cornercopoint,shorten right=5pt,kpoint common]
\tikzstyle{kpoint symm}=[shape=cornerpoint,shorten left=5pt,shorten right=5pt,kpoint common]

\tikzstyle{wide kpoint sc}=[shape=cornerpoint,kpoint common, minimum width=1 cm]
\tikzstyle{wide kpointdag sc}=[shape=cornercopoint,kpoint common, minimum width=1 cm]

\tikzstyle{black kpoint}=[shape=cornerpoint,shorten left=5pt,kpoint common,fill=black,font=\color{white}]

\tikzstyle{black kpoint sm}=[shape=cornerpoint,shorten left=5pt,kpoint common,fill=black,font=\color{white},scale=0.75]

\tikzstyle{black kpoint adjoint}=[shape=cornercopoint,shorten left=5pt,kpoint common,fill=black,font=\color{white}]
\tikzstyle{black kpointadj}=[shape=cornercopoint,shorten left=5pt,kpoint common,fill=black,font=\color{white}]

\tikzstyle{black kpointadj sm}=[shape=cornercopoint,shorten left=5pt,kpoint common,fill=black,font=\color{white},scale=0.75]

\tikzstyle{black dkpoint}=[shape=cornerpoint,shorten left=5pt,kpoint common,fill=black, doubled,font=\color{white}]
\tikzstyle{black dkpoint adjoint}=[shape=cornercopoint,shorten left=5pt,kpoint common,fill=black, doubled,font=\color{white}]
\tikzstyle{black dkpointadj}=[shape=cornercopoint,shorten left=5pt,kpoint common,fill=black, doubled,font=\color{white}]

\tikzstyle{black dkpoint sm}=[shape=cornerpoint,shorten left=5pt,kpoint common,fill=black, doubled,font=\color{white},scale=0.75]
\tikzstyle{black dkpointadj sm}=[shape=cornercopoint,shorten left=5pt,kpoint common,fill=black, doubled,font=\color{white},scale=0.75]

\tikzstyle{kpointdag}=[kpoint adjoint]
\tikzstyle{kpointadj}=[kpoint adjoint]
\tikzstyle{kpointconj}=[kpoint conjugate]
\tikzstyle{kpointtrans}=[kpoint transpose]

\tikzstyle{big kpoint}=[kpoint, minimum width=1.2 cm, minimum height=8mm, inner sep=4pt, text depth=3mm]

\tikzstyle{wide kpoint}=[kpoint, minimum width=1 cm, inner sep=2pt]
\tikzstyle{wide kpointdag}=[kpointdag, minimum width=1 cm, inner sep=2pt]
\tikzstyle{wide kpointconj}=[kpointconj, minimum width=1 cm, inner sep=2pt]
\tikzstyle{wide kpointtrans}=[kpointtrans, minimum width=1 cm, inner sep=2pt]

\tikzstyle{wider kpoint}=[kpoint, minimum width=1.25 cm, inner sep=2pt]
\tikzstyle{wider kpointdag}=[kpointdag, minimum width=1.25 cm, inner sep=2pt]
\tikzstyle{wider kpointconj}=[kpointconj, minimum width=1.25 cm, inner sep=2pt]
\tikzstyle{wider kpointtrans}=[kpointtrans, minimum width=1.25 cm, inner sep=2pt]

\tikzstyle{gray kpoint}=[kpoint,fill=gray!50!white]
\tikzstyle{gray kpointdag}=[kpointdag,fill=gray!50!white]
\tikzstyle{gray kpointadj}=[kpointadj,fill=gray!50!white]
\tikzstyle{gray kpointconj}=[kpointconj,fill=gray!50!white]
\tikzstyle{gray kpointtrans}=[kpointtrans,fill=gray!50!white]

\tikzstyle{gray dkpoint}=[kpoint,fill=gray!50!white,doubled]
\tikzstyle{gray dkpointdag}=[kpointdag,fill=gray!50!white,doubled]
\tikzstyle{gray dkpointadj}=[kpointadj,fill=gray!50!white,doubled]
\tikzstyle{gray dkpointconj}=[kpointconj,fill=gray!50!white,doubled]
\tikzstyle{gray dkpointtrans}=[kpointtrans,fill=gray!50!white,doubled]

\tikzstyle{white label}=[draw,fill=white,rectangle,inner sep=0.7 mm]
\tikzstyle{gray label}=[draw,fill=gray!50!white,rectangle,inner sep=0.7 mm]
\tikzstyle{black label}=[draw,fill=black,rectangle,inner sep=0.7 mm]

\tikzstyle{dkpoint}=[kpoint,doubled]
\tikzstyle{wide dkpoint}=[wide kpoint,doubled]
\tikzstyle{dkpointdag}=[kpoint adjoint,doubled]
\tikzstyle{wide dkpointdag}=[wide kpointdag,doubled]
\tikzstyle{dkcopoint}=[kpoint adjoint,doubled]
\tikzstyle{dkpointadj}=[kpoint adjoint,doubled]
\tikzstyle{dkpointconj}=[kpoint conjugate,doubled]
\tikzstyle{dkpointtrans}=[kpoint transpose,doubled]

\tikzstyle{kscalar}=[kpoint common, shape=EBox, inner xsep=-1pt, inner ysep=3pt,font=\small]
\tikzstyle{kscalarconj}=[kpoint common, shape=WBox, inner xsep=-1pt, inner ysep=3pt,font=\small]

\tikzstyle{spekpoint}=[kpoint sc,minimum height=5mm,inner sep=3pt]
\tikzstyle{spekcopoint}=[kpoint adjoint sc,minimum height=5mm,inner sep=3pt]

\tikzstyle{dspekpoint}=[spekpoint,doubled]
\tikzstyle{dspekcopoint}=[spekcopoint,doubled]

 \tikzstyle{upground}=[circuit ee IEC,thick,ground,rotate=90,scale=2.5]
 \tikzstyle{downground}=[circuit ee IEC,thick,ground,rotate=-90,scale=2.5]
 \tikzstyle{infupground}=[circuit ee IEC,thick,ground,rotate=0,scale=2.5]
 \tikzstyle{infdownground}=[circuit ee IEC,thick,ground,rotate=180,scale=2.5]
 \tikzstyle{bigground}=[regular polygon,regular polygon sides=3,draw=gray,scale=0.50,inner sep=-0.5pt,minimum width=10mm,fill=gray]

\tikzstyle{arrs}=[-latex,font=\small,auto]
\tikzstyle{arrow plain}=[arrs]
\tikzstyle{arrow dashed}=[dashed,arrs]
\tikzstyle{arrow bold}=[very thick,arrs]
\tikzstyle{arrow hide}=[draw=white!0,-]
\tikzstyle{arrow reverse}=[latex-]
\tikzstyle{cdnode}=[]

\let\olddagger\dagger
\renewcommand{\dagger}{\ensuremath{\olddagger}\xspace}

\usepackage{makeidx}
\makeindex

\usepackage{color}
\def\bR{\begin{color}{red}}
\def\bB{\begin{color}{blue}}
\def\bM{\begin{color}{magenta}}
\def\bC{\begin{color}{cyan}}
\def\bW{\begin{color}{white}}
\def\bBl{\begin{color}{black}}
\def\bG{\begin{color}{green}}
\def\bY{\begin{color}{yellow}}
\def\e{\end{color}\xspace}
\newcommand{\bit}{\begin{itemize}}
\newcommand{\eit}{\end{itemize}\par\noindent}
\newcommand{\ben}{\begin{enumerate}}
\newcommand{\een}{\end{enumerate}\par\noindent}
\newcommand{\beq}{\begin{equation}}
\newcommand{\eeq}{\end{equation}\par\noindent}
\newcommand{\beqa}{\begin{eqnarray*}}
\newcommand{\eeqa}{\end{eqnarray*}\par\noindent}
\newcommand{\beqn}{\begin{eqnarray}}
\newcommand{\eeqn}{\end{eqnarray}\par\noindent}

\def\jR{\begin{color}{black}}
\def\jB{\begin{color}{black}}
\def\jM{\begin{color}{magenta}}
\def\jC{\begin{color}{cyan}}
\def\jW{\begin{color}{white}}
\def\jBl{\begin{color}{black}}
\def\jG{\begin{color}{green}}
\def\jY{\begin{color}{yellow}}

\begin{document}

\title{A hierarchy of semidefinite programs for generalised Einstein-Podolsky-Rosen scenarios}

\author{Matty J.~Hoban}
\affiliation{Quantum Group, Department of Computer Science, University of Oxford, United Kingdom}
\author{Tom Drescher}
\affiliation{University of Innsbruck, Department of Mathematics,
A-6020 Innsbruck, Austria}
\author{Ana Bel\'en Sainz}
\affiliation{International Centre for Theory of Quantum Technologies, University of  Gda{\'n}sk, 80-309 Gda{\'n}sk, Poland}
\affiliation{Basic Research Community for Physics e.V., Germany}

\begin{abstract}
Correlations in Einstein-Podolsky-Rosen (EPR) scenarios, captured by \textit{assemblages} of unnormalised quantum states, have recently caught the attention of the community, both from a foundational and an information-theoretic perspective. The set of quantum-realisable assemblages, or abbreviated to quantum assemblages, are those that arise from multiple parties performing local measurements on a shared quantum system. In general, deciding whether or not a given assemblage is a quantum assemblage, i.e. membership of the set of quantum assemblages, is a hard problem, and not always solvable. In this paper we introduce a hierarchy of tests where each level either determines non-membership of the set of quantum assemblages or is inconclusive. The higher the level of the hierarchy the better one can determine non-membership, and this hierarchy converges to a particular set of assemblages. Furthermore, this set to which it converges contains the quantum assemblages. Each test in the hierarchy is formulated as a semidefinite program. This hierarchy allows one to upper bound the quantum violation of a steering inequality and the quantum advantage provided by quantum EPR assemblages in a communication or information-processing task. 
\end{abstract}

\maketitle

\newpage

\tableofcontents

\newpage

\section{Introduction}

Einstein-Podolsky-Rosen (EPR) `steering' is a remarkable nonclassical feature of quantum theory \cite{einstein1935can,wiseman2007steering}, first discussed by Schr\"odinger \cite{schrodinger1935discussion}. 
Some people understand EPR nonclassicality as a phenomenon where Alice, by performing measurements on half of a shared system, seemingly remotely ‘steers’ the quantum state of a distant Bob, in a way which has no classical explanation. A different understanding of this -- endorsed by EPR themselves –- is that Alice has no causal influence on the
physical state of Bob’s system, and just merely updates her knowledge of the state of Bob’s system by performing a measurement on a system correlated with his. 
Regardless of how one should view the EPR phenomenon, the correlations that Alice and Bob observe between her measurement choices \& outcomes and the state of his system are relevant not only for exploring the foundations of nature but also for the quantum information community: on the one hand,  EPR scenarios show a particular way in which the world is nonclassical, opening up new pathways toward pursuing an understanding of nature itself. On the other hand, these correlations in EPR scenarios enable the certification of entanglement under relaxed assumptions \cite{cavalcanti2015detection,mattar2017experimental}, and constitute an information-theoretic resource for various cryptographic tasks \cite{branciard2012one,gianissecretsharing}. These correlations in an EPR steering experiment are usually mathematically captured by an \textit{assemblage} of possibly unnormalised quantum states, each of which representing the state of Bob's quantum system for each of Alice's measurement choice \& outcome \cite{pusey2013negativity}.

A recent line of research explores the possible assemblages that hypothetical theories which supersede quantum theory may allow. This has been pursued from the perspective of ``constraining the possible assemblages only by the operational constraints of the EPR scenario'' \cite{gisin1989stochastic,hughston1993complete,pqs,bpqs}, and also from the perspective of identifying the assemblages that certain toy theories allow \cite{cavalcanti2021witworld}. In both cases, the minimal constraint that the assemblages need to satisfy is the so-called No Signalling principle, which for the case of traditional bipartite EPR scenarios reads as follows. Let $\sigma_{a|x}$ denote the assemblage elements, where $a$ denotes Alice's outcome, and $x$ denotes her choice of measurement. The No Signalling principle demands that $\sum_{a} \sigma_{a|x} = \sum_{a} \sigma_{a|x'}$ for all possible pairs $(x,x')$, that is, Alice can not transfer information to Bob by leveraging her choice of measurement. It is known that there exist assemblages that satisfy the No Signalling principle, however cannot be produced by performing a quantum EPR experiment, that is, do not admit a \textit{quantum explanation} \cite{pqs,bpqs}. The phenomenon of non-signalling assemblages which do not admit a quantum explanation is referred to as \textit{post-quantum steering} or \textit{post-quantum EPR inference}. 

The logical possibility of post-quantum steering has then opened a few questions: if Alice and Bob observe an assemblage, how can they check if it has a quantum explanation without making various assumptions about their experimental setup? This question has been studied at length for Bell experiments and Kochen-Specker contextuality experiments, and suspected to be not possible to answer in full generality. Hence, tools have been developed to certify when correlations in Bell or Kochen-Specker ``definitely do not admit of a quantum explanation''. These tools come as hierarchies of tests, each defined as a semidefinite program: testing one level tells you either ``it does not have a quantum explanation'' or ``it's unclear if it has or not a quantum explanation''. If the level you're testing tells you the second answer, then you move on to a higher level of the hierarchy, and so on. For the case of Bell scenarios, the hierarchy of semidefinite tests was developed by Navascu\'es, Pironio, and Ac\'in \cite{NPA}, and for the case of Kochen-Specker scenarios it was developed by Ac\'in, Fritz, Leverrier, and Sainz \cite{acin2015combinatorial}. These hierarchies are in practice essential when assessing the best performance of quantum correlations in communication and information processing tasks and even for upper-bounding the maximum quantum violation of a Bell or Kochen-Specker inequality \cite{BellRev}. 

For the case of EPR scenarios, one hierarchy of semidefinite tests was developed by Johnston, Mittal, Russo, and Watrous \cite{ENG}, in the particular case of a traditional tripartite EPR scenario with two untrusted parties. This technique was applied in the context of the so-called \textit{extended nonlocal games}, to bound the maximum chance of success when playing these games with quantum resources. In our work we develop a new hierarchy of tests (defined each as a semidefinite program) that can check whether an assemblage has no quantum realisation in all the EPR scenarios we know of so far, such as traditional multipartite EPR scenarios and generalised bipartite scenarios, including Bob-with-Input scenarios and Instrumental EPR scenarios. We show that this hierarchy converges to a particular set of assemblages, but we leave it as an open question whether there is a natural definition of this set. Our hierarchy has already proven useful for identifying post-quantum assemblages in recent work \cite{bpqs}.

\section{Bipartite Bob-with-Input EPR scenario}

We begin by introducing our hierarchy for semidefinite tests for the simplest case of a so-called Bob-with-Input scenario: that with only one so-called black-box party in the setup. First, we recap the main concepts for the scenario, and then we define the hierarchy of assemblages sets. We show that quantum assemblages belong to the set of assemblages defined by any level of this hierarchy, and moreover show that this is a convergent hierarchy. 

\subsection{Definition of this EPR scenario}\label{se:bwi}

The particular type of bipartite EPR scenario that we focus on here is called \textit{Bob-with-Input} \cite{bpqs}, which is defined as follows. Just like in a traditional EPR scenario, two distant parties (Alice and Bob) share a physical system. Alice performs a measurement (chosen from a set) on her share of the system, in a space-like separated way from Bob. When obtaining the outcome, Alice updates and refines her knowledge of the state of Bob's subsystem. The additional feature of this EPR scenario, relative to the traditional one, is that here Bob can moreover influence the effective state preparation of his quantum system by choosing (from a finite set) the value of a classical variable $y$. This situation is depicted in Figure \ref{fig:steebip}. 

\begin{figure}
\begin{center}
	    \begin{tikzpicture}[scale=1.2]
		\node at (-2,1.3) {Alice};
		\shade[draw, thick, ,rounded corners, inner color=white,outer color=gray!50!white] (-1.7,-0.3) rectangle (-2.3,0.3) ;
		\draw[thick, ->] (-2.5,0.5) to [out=180, in=90] (-2,0.3);
		\node at (-2.7,0.5) {$x$};
		\draw[thick, -<] (-2,-0.3) to [out=-90, in=180] (-2.5,-0.5);
		\node at (-2.8,-0.5) {$a$};

		\node at (2,1.3) {Bob};
		\shade[draw, thick, ,rounded corners, inner color=white,outer color=gray!50!white] (1.7,-0.3) rectangle (2.3,0.3) ;		
		\draw[thick, ->] (2.5,0.5) to [out=180, in=90] (2,0.3);
		\node at (2.7,0.5) {$y$};
		\draw[thick, ->] (2,-0.3) -- (2,-0.7);
		\node at (2,-1) {$\sigma_{a|xy}$};
		
		\node at (0,0) {s};
		\draw[thick, dashed, color=gray!70!white] (0,0) circle (0.3cm);
		\draw[thick, dashed, color=gray!70!white, ->] (-0.3,0) -- (-1.6,0);
		\draw[thick, dashed, color=gray!70!white, ->] (0.3,0) -- (1.6,0);		
	    \end{tikzpicture}	    
\end{center}
\caption{Bob-with-Input steering scenario, where Bob now has an input, which is used to determine the production of a quantum system.}
\label{fig:steebip}
\end{figure}
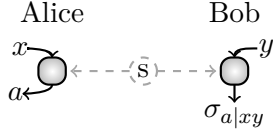

Formally, let $x\in \{1, \ldots, m_A\} =: \X$ denote the choice of Alice's measurement, $a \in \{0, \ldots, o_A-1\} =: \A$ denote the outcome of Alice's measurement\footnote{In principle, the measurements may not need to have all the same number of outcomes, although for the purpose of this manuscript one may assume that is indeed the case without loss of generality \cite{BellRev}.}, $y \in \{1, \ldots, m_B\} =: \Y$ be the choice of Bob's input, and $d$ the dimension of the Hilbert space of the quantum system whose state is prepared by this protocol. Then, the collection of states that Bob's system may be prepared into, called \textit{assemblage} and denoted by $\As$, consists of the following elements: $\{\sigma_{a|xy}\}_{a \in \A, \, x \in \X, \, y \in \Y} \equiv \As$. 

Notice that this experiment is not assumed to be describable by quantum theory in full generality. Indeed, the system that Alice and Bob share may as well belong to a (hypothetical) physical theory that supersedes quantum. The only constraint here is that the state of the system after the protocol is performed corresponds to a quantum state\footnote{Indeed, in traditional EPR scenarios, the staring point is that Bob's system admits an effective quantum description, even if the joint system shared with Alice may be outside quantum theory.}. 

Now, depending on what the nature of the system shared by Alice and Bob is, the assemblage may (i) be classically explainable -- called \textit{LHS Bob-with-Input assemblages}, (ii) be quantumly explainable -- called \textit{Quantum Bob-with-Input assemblages}, or (iii) simply comply with the operational constraints imposed by the mere definition of the steering scenario: that the assemblage is compatible with a common cause\footnote{Originally, an assemblage was defined to comply with the operational constraints of the steering scenario if it satisfied the No Signalling principle between the two parties. This has now been shown to be equivalent to demanding that the assemblage has a common cause explanation within a General Probabilistic Theory that supersedes quantum \cite{cavalcanti2023every}. Here we adopt the latter viewpoint, for the arguments presented in Refs.~\cite{Cowpie, DavidBwI, rtLOSE, anom}.} -- called \textit{Non-signalling Bob-with-Input assemblages} (see Sec.~\ref{se:bwiMul} for the specification of the so-called non-signalling constraints which define these assemblages).

This paper is motivated by the question of how to certify that an assemblage is (or rather, is not) quantumly explainable. Following the ideas of the Navascu\'es-Pironio-Ac\'in hierarchy \cite{NPA} for Bell Scenarios, and the Steering-hierarchy for traditional steering scenarios \cite{ENG}, we define a collection of sets of Non-signalling assemblages, which satisfy the following properties: 
\begin{compactitem}
\item the collection forms a hierarchy, in the sense that one may order the sets so that each is contained within all the sets that precedes it in the order,
\item testing membership to each set is one instance of a semidefinite program (SDP), 
\item the set of quantum assemblages is included within all the sets in the hierarchy. 
\end{compactitem}
Similarly to the hierarchies that have been developed for other scenarios \cite{ENG}, our method serves as a useful tool to certify post-quantum assemblages, i.e., to certify when an assemblage may not be quantumly explainable. In addition, our method also poses as a useful tool to bound the maximum quantum violation of steering inequalities in the Bob-with-Input scenario. 

In the following sections we will define the hierarchy of sets, show that quantum assemblages belong to all of them, and prove that the hierarchy converges.

\subsection{Definition of the hierarchy}\label{se:BwIhier}

Here we define the hierarchy of sets of assemblages in the Bob-with-Input scenario that we motivated previously. 
\blk The definition of each set of assemblages will be rather mathematical and rely on the specification of (mathematical rather than physical) constraints that the assemblage elements must satisfy, where these constraints will be specified by demanding the existence of peculiar matrices with particular properties. \blk The main ingredients in the definition of \blk (the set of assemblages for) \blk each level are:
\begin{compactitem}
\item the definition of a set of words,
\item the specification of equivalence relations between those words, 
\item the definition of a moment matrix.
\end{compactitem}

\blk
Broadly speaking, the \textit{set of words} will help us index the rows and columns of the so-called \textit{moment matrix} (which we will define later on), and the \textit{equivalence relations between words} will allow us to specify which entries of the moment matrix shoulc be equal to one another. 
\blk 

The alphabet $\Upsilon$ that we will use to define the words is composed of the following elements (letters\footnote{{Strictly speaking, $\emptyset$ is not a letter but rather an \textit{empty string/word}, although for simplicity in the narrative we will not always insist on this jargon.}}): 
\begin{align}
\Upsilon := \{\emptyset\} \cup \{ a|x \}_{{ x \in \X\,, a \in \A \setminus \{0\}}} \cup \{ y \}_{y \in \Y} \,. 
\end{align}

A \textit{word} then is a concatenation of elements drawn from $\Upsilon$. Examples of words are $1|x \, \emptyset \, y$, $y \, \emptyset \, y' \, 1|x$, $a|x \, y \, 1|x' $, and $a|xy$. Hereon we will denote words by the symbols $\mathbf{v}, \mathbf{w}, \mathbf{u}$; $\mathbf{v}\mathbf{w}$ will denote the concatenation of the two words (which forms a new word), and $\mathbf{v}^\dagger$ will denote the word defined by the letters of $\mathbf{v}$ written in reverse order. Finally, we denote by $\cS^*$ the set of all words of arbitrary length with letters drawn from $\Upsilon$.

To determine equivalence relations between words, denoted by the symbol $\equiv$, we first need to establish the symmetry operations on the elements of the alphabet:
\begin{compactitem} 
\item $\mathbf{v}\mathbf{w} \equiv \mathbf{v} \emptyset \mathbf{w} $ for all $\mathbf{u},\mathbf{w} \in \cS^*$,
\item $\mathbf{v}\mathbf{v} \equiv \mathbf{v}$ for all $\mathbf{v} \in \Upsilon$,
\item $a|xy \equiv y a|x$ for all $y,a|x \in \Upsilon$.
\end{compactitem} 
We say then that $\mathbf{v} \equiv \mathbf{w}$ if there exists a sequence of symmetry operations that can take $\mathbf{v}$ to $\mathbf{w}$ (and vice-versa). 
For example, $y \, \emptyset \, a|x \equiv a|x \, y$ and $y \, \emptyset \, y' \, 1|x \equiv 1|x \, y \, y'$, but $a|x \, y \, y' \not\equiv y' \, a|x \, y$ when $y \neq y'$. 

Moreover, we say that a word $\mathbf{v}$ is \textit{null} if, after applying a sequence of symmetry operations, one may find a letter $a|x$ followed by a letter $a'|x$ for some $x \in \X$ and $a \neq a' \in \A \setminus \{0\}$. For example, the word $1|x \, y \, 2|x$ is null, whereas the word $1|x \, y \, 2|x'$ with $x \neq x'$ is not.

The next ingredient we need to set up before we can define the hierarchy is the following:
\begin{defn} \textbf{The set of words} $\boldsymbol{\cS_n}$.--\\
A word $\mathbf{v} \in \cS^*$ belongs to $\cS_n$ if it may arise from concatenating at most $n$ letters of the alphabet $\Upsilon$.
\end{defn}
\noindent For example, $a|x \in \cS_3$, while $a|x \, y \, a|x' \not\in \cS_2$ when $x \neq x'$. By $\lvert \cS_n \lvert$ we will denote the cardinality of $\cS_n$. \blk The idea is that, for each $n$, the words in $\cS_n$ will index the rows and columns for a different moment matrix, as we see next. \blk

Now we can define the concept of a moment matrix within the context of our hierarchy:
\begin{defn}\label{def:MMn}\textbf{Moment matrix of order $\mathbf{n}$: $\boldsymbol{\Gamma_n}$}.--\\
Let $\Gamma_n$ be a matrix of size $\lvert \cS_n \lvert \times \lvert \cS_n \lvert$, whose entries are $d \times d$ complex matrices, and whose rows and columns are indexed by the words in $\cS_n$. This matrix is a moment matrix of order $n$ iff it satisfies the following properties: 
\begin{align}
\Gamma_n \geq 0 \,, \label{eq:2}\\
\Gamma_n(\emptyset,\emptyset) &= \id_d \,,\label{eq:3}\\
\Gamma_n(\mathbf{v},\mathbf{w}) &= \Gamma_n(\mathbf{v'},\mathbf{w'}) \quad \text{if} \quad \mathbf{v}^\dagger\mathbf{w} \equiv \mathbf{v'}^\dagger\mathbf{w'} \,,\label{eq:4}\\
\Gamma_n(\mathbf{v},\mathbf{w}) &= \mathbf{0}_d  \quad \text{if} \quad \mathbf{v}^\dagger\mathbf{w}  \, \text{is null} \,,\label{eq:5}\\
\Gamma_n(a|x,a'|x') &\propto \id_d \quad \text{for all} \quad a|x,a'|x' \in \Upsilon \,, \label{eq:6}
\end{align}
where $\mathbf{0}_d$ is the $d \times d$ matrix whose all entries are 0. 
\end{defn}

 Let us now bring in the connection with assemblages in the Bob-with-Input scenario. A moment matrix $\Gamma_n$ is a \textit{certificate of order $n$} for the assemblage $\{\sigma_{a|xy}\}_{a \in \A, \, x \in \X, \, y \in \Y}$ when the following conditions are met: 
\begin{defn}\label{def:cert} \textbf{Certificate of order $n$ for an assemblage $\As$}.--\\
Let $\As = \{\sigma_{a|xy}\}_{a \in \A, \, x \in \X, \, y \in \Y}$ be a Non-signalling assemblage in the Bob-with-Input EPR scenario. 
A matrix $\Gamma_n$ is a certificate of order $n$ for $\As$ iff $\Gamma_n$ is a moment matrix of order $n$ and the following are satisfied: 
\begin{align}
\Gamma_n(\emptyset,a|x) &= \Tr{\sigma_{a|x1}} \, \id_d \quad \forall \, a|x \in \Upsilon\,, \label{eq:7}\\
\Gamma_n(\emptyset,y) &= \frac{1}{d} \, \sigma^\mathrm{T}_{y} \quad \forall \, y \in \Upsilon \,,\label{eq:8}\\
\Gamma_n(\emptyset, a|x \, y) &= \frac{1}{d} \, \sigma^\mathrm{T}_{a|xy} \quad \forall a|x,y \in \Upsilon\,,\label{eq:9}
\end{align}
where $\sigma_y = \sum_{a\in\A} \sigma_{a|xy}$ is the marginal state on Bob's side, and $(\cdot)^\mathrm{T}$ denotes transposition.
\end{defn}

\blk
This definition for a Certificate of order $n$ in a Bob-with-Input EPR scenario is inspired by those found in the literature for other phenomena (see, e.g., Refs.~\cite{acin2015combinatorial,fritz2013probabilistic} for Kochen-Speker contextuality). The idea is that the first row of the matrix $\Gamma_n$ contains all the information available from the assemblage elements $\{\sigma_{a|xy}\}$, and moreover, that the information be displayed in a way that is readily accessible in a systematic way. For example, Eq.~\eqref{eq:7} stipulates which matrix elements provide the information on the probabilities $p(a|x)$ alone, where the operator $\id_d$ is needed to satisfy the fact that the entries of $\Gamma_n$ are themselves $d$-dimensional matrices. Eq.~\eqref{eq:9} and Eq.~\eqref{eq:8} tell you how to relate the entries $a|xy$ and $y$ with the assemblage elements $\sigma_{a|xy}$ and their marginals $\sigma_y$ respectively: notice that the relationship is not direct but that the operators are transposed and re-normalised -- this is a necessary technicality to be able to easily construct a quantum model from the elements of $\Gamma_n$ later on. 
\blk

Now we can finally define the sets of assemblages that form our hierarchy. We will denote these sets by $\cQ^n$:
\begin{defn}\textbf{Set of assemblages in level $n$ of the hierarchy: $\cQ^n$}.--\\
An assemblage $\As$ in the Bob-with-Input EPR scenario belongs to the set $\cQ^n$ iff there exists a matrix $\Gamma_n$ that is a certificate of order $n$ for $\As$. 
\end{defn}

It is easy to see that $\As \in \cQ^n \, \Rightarrow \, \As \in \cQ^{n-1}$: this follows from the fact that one can take the certificate of order $n$ for $\As$, and construct from it a new matrix $\Gamma_{n-1}$ by keeping the rows/columns of $\Gamma_n$ which are labelled by the words of $\cS_{n-1}$. By definition, $\Gamma_{n-1}$ will be a moment matrix and satisfy the conditions to be a certificate of order $n-1$ for $\As$. It hence follows that $\cQ^n \subseteq \cQ^{n-1}$ for all $n>1$. 

Next we will show that quantum assemblages in the Bob-with-Input scenario belong to $\cQ^n$ for all $n$. 

\subsection{Quantum assemblages are $\cQ^n$ assemblages}\label{se:3}

In this section we will show that each level of the hierarchy is itself a relaxation of the set of quantum assemblages in the Bob-with-Input scenario. The main idea behind this proof is to, starting from a quantum assemblage, define a state and operators with which to construct a moment matrix $\Gamma_n$ that serves as a certificate of order $n$ for the original assemblage. 

Some of the construction will be presented in the diagrammatic language of Ref.~\cite{PT} for simplicity. Consider a quantum assemblage $\As_\mathrm{Q} = \{\sigma_{a|xy}\}_{a \in \A, \, x \in \X, \, y \in \Y}$, and let $\rho$, $\{\Pi_{a|x}\}_{a \in \A,\, x \in \X}$, and $\{\cE^y\}_{y \in \Y}$ be a state, measurements on Alice's system, and \blk completely-positive trace-reserving (CPTP) \blk  maps\footnote{\blk A linear map $\Lambda$ is \textit{positive} if it maps positive elements into positive elements. A linear map $\Lambda$ is in addition \textit{completely-positive} if the linear map $\Lambda \otimes \mathbb{I}$ is a positive map for any value of the dimension of the Hilbert space the identity operator $\mathbb{I}$ acts on. Finally, a linear operator $\Lambda$ is \textit{trace-preserving} if, for each element $v$ in its domain, the trace of the image of the element is equal to the trace of the element, i.e., $\Tr{\Lambda[v]}=\Tr{v}$. \blk} on Bob's system, respectively, that provide a realisation of the assemblage. That is,
\begin{align*}
\sigma_{a|xy} \quad = \quad %
\begin{tikzpicture}
	\begin{pgfonlayer}{nodelayer}
		\node [style=none] (0) at (-0.5, -1) {$\rho$};
		\node [style=none] (1) at (-4, -0.5) {};
		\node [style=none] (2) at (-0.5, -2) {};
		\node [style=none] (3) at (3, -0.5) {};
		\node [style=none] (4) at (-2.75, 1.5) {\tiny{$\Pi_{a|x}$}};
		\node [style=none] (5) at (-2.75, -0.5) {};
		\node [style=small box] (6) at (2, 1.75) {$\mathcal{E}^{y}$};
		\node [style=none] (7) at (2, -0.5) {};
		\node [style=none] (8) at (2, 3.25) {};
		\node [style=right label] (9) at (-2.75, 0.25) {$A$};
		\node [style=right label] (10) at (2, 0.25) {$B_{(i)}$};
		\node [style=right label] (11) at (2, 3) {$B$};
		\node [style=none] (12) at (-4, 1) {};
		\node [style=none] (13) at (-1.5, 1) {};
		\node [style=none] (14) at (-2.75, 2.75) {};
		\node [style=none] (15) at (-2.75, 1) {};
	\end{pgfonlayer}
	\begin{pgfonlayer}{edgelayer}
		\draw (1.center) to (2.center);
		\draw (2.center) to (3.center);
		\draw (3.center) to (1.center);
		\draw [qWire] (6) to (7.center);
		\draw [qWire] (8.center) to (6);
		\draw (14.center) to (12.center);
		\draw (12.center) to (13.center);
		\draw (13.center) to (14.center);
		\draw [qWire] (15.center) to (5.center);
	\end{pgfonlayer}
\end{tikzpicture}
}.
\end{align*}
Notice that, because of the Stinespring dilation theorem  \cite{vp,stine}, we can always take Alice's measurements to be projective and the state $\rho$ to be pure, without loss of generality. 

The next step is to realise that one may always find an equivalent realisation of the assemblage where Bob's transformations are unitary ones, by dilating the CPTP maps $\{\cE^y\}_{y \in \Y}$ and considering the corresponding auxiliary system as part of the original preparation procedure. That is, 
\begin{align}\label{eq:therhop}
\sigma_{a|xy} \quad = \quad %
\begin{tikzpicture}
	\begin{pgfonlayer}{nodelayer}
		\node [style=none] (0) at (-0.5, -1) {$\rho^\prime$};
		\node [style=none] (1) at (-4, -0.5) {};
		\node [style=none] (2) at (-0.5, -2) {};
		\node [style=none] (3) at (3, -0.5) {};
		\node [style=none] (4) at (-2.75, 1.5) {\tiny{$\Pi_{a|x}$}};
		\node [style=none] (5) at (-2.75, -0.5) {};
		\node [style=none] (7) at (2.25, -0.5) {};
		\node [style=none] (8) at (2.25, 3.25) {};
		\node [style=right label] (9) at (-2.75, 0.25) {$A$};
		\node [style=right label] (10) at (2.25, 0.25) {$B_{(i)}$};
		\node [style=right label] (11) at (2.25, 2.75) {$B$};
		\node [style=none] (12) at (-4, 1) {};
		\node [style=none] (13) at (-1.5, 1) {};
		\node [style=none] (14) at (-2.75, 2.75) {};
		\node [style=none] (15) at (-2.75, 1) {};
		\node [style=none] (16) at (0.5, 1) {};
		\node [style=none] (17) at (3, 1) {};
		\node [style=none] (18) at (3, 2) {};
		\node [style=none] (19) at (0.5, 2) {};
		\node [style=none] (20) at (1.75, 1.5) {$U_y$};
		\node [style=none] (21) at (2.25, 2) {};
		\node [style=none] (22) at (2.25, 1) {};
		\node [style=none] (23) at (1.25, -0.5) {};
		\node [style=left label] (24) at (1.25, 0.25) {$B_{aux}$};
		\node [style=none] (25) at (1.25, 1) {};
		\node [style=none] (26) at (1.25, 2.5) {};
		\node [style=none] (28) at (1.25, 2) {};
		\node [style=upground] (29) at (1.25, 2.75) {};
	\end{pgfonlayer}
	\begin{pgfonlayer}{edgelayer}
		\draw (1.center) to (2.center);
		\draw (2.center) to (3.center);
		\draw (3.center) to (1.center);
		\draw (14.center) to (12.center);
		\draw (12.center) to (13.center);
		\draw (13.center) to (14.center);
		\draw [qWire] (15.center) to (5.center);
		\draw (19.center) to (16.center);
		\draw (16.center) to (17.center);
		\draw (17.center) to (18.center);
		\draw (18.center) to (19.center);
		\draw [qWire] (21.center) to (8.center);
		\draw [qWire] (22.center) to (7.center);
		\draw [qWire] (25.center) to (23.center);
		\draw [qWire] (28.center) to (26.center);
	\end{pgfonlayer}
\end{tikzpicture}
},
\end{align}
where $\rho^\prime = \rho \otimes \rho_{aux}$ may always be taken to be a pure state, without loss of generality. 

Now, consider another system in Bob's lab, called $B'$, of the same dimension as $B$. Let $\{ \Phi_1, \Phi_2\}$ be the joint projective measurement on $BB'$, of two outcomes, defined as follows: 
\begin{align*}
\begin{cases}
\Phi_1 &= \ket{\phi^+}\bra{\phi^+} \,,\\
\Phi_2 &= \id - \Phi_1\,,
\end{cases}
\end{align*}
where $\ket{\phi^+} = \frac{1}{\sqrt{d}} \sum_{j=1:d} \ket{j}_B\otimes \ket{j}_{B'}$ is the LOCC-maximally-entangled\footnote{Where LOCC stands for Local Operations and Classical Communication.} state for $BB'$ \cite{anom,Cowpie}.

One may now define a family of dichotomic measurements $\{\Psi_{b|y}\}_{b \in \{1,2\},\, y \in \Y}$ on the composite system $B_{aux}\cdot B_{(i)}\cdot B'$ as follows: 
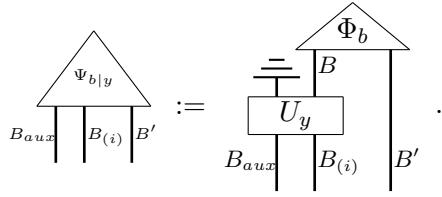
\begin{align*}
\begin{tikzpicture}
	\begin{pgfonlayer}{nodelayer}
		\node [style=none] (4) at (-2.5, 0.75) {\tiny{$\Psi_{b|y}$}};
		\node [style=none] (5) at (-3.5, -1.5) {};
		\node [style=left label] (9) at (-3.5, -0.75) {\tiny{$B_{aux}$}};
		\node [style=none] (12) at (-4, 0) {};
		\node [style=none] (13) at (-1, 0) {};
		\node [style=none] (14) at (-2.5, 2) {};
		\node [style=none] (15) at (-3.5, 0) {};
		\node [style=none] (16) at (-2.75, -1.5) {};
		\node [style=right label] (17) at (-2.75, -0.75) {\tiny{$B_{(i)}$}};
		\node [style=none] (18) at (-2.75, 0) {};
		\node [style=none] (19) at (-1.5, -1.5) {};
		\node [style=right label] (20) at (-1.5, -0.75) {\tiny{$B'$}};
		\node [style=none] (21) at (-1.5, 0) {};
	\end{pgfonlayer}
	\begin{pgfonlayer}{edgelayer}
		\draw (14.center) to (12.center);
		\draw (12.center) to (13.center);
		\draw (13.center) to (14.center);
		\draw [qWire] (15.center) to (5.center);
		\draw [qWire] (18.center) to (16.center);
		\draw [qWire] (21.center) to (19.center);
	\end{pgfonlayer}
\end{tikzpicture}
} := %
\begin{tikzpicture}
	\begin{pgfonlayer}{nodelayer}
		\node [style=none] (7) at (2.25, -2.25) {};
		\node [style=none] (8) at (2.25, 1.5) {};
		\node [style=right label] (10) at (2.25, -1.5) {$B_{(i)}$};
		\node [style=right label] (11) at (2.25, 1) {$B$};
		\node [style=none] (16) at (0.5, -0.75) {};
		\node [style=none] (17) at (3, -0.75) {};
		\node [style=none] (18) at (3, 0.25) {};
		\node [style=none] (19) at (0.5, 0.25) {};
		\node [style=none] (20) at (1.75, -0.25) {$U_y$};
		\node [style=none] (21) at (2.25, 0.25) {};
		\node [style=none] (22) at (2.25, -0.75) {};
		\node [style=none] (23) at (1.25, -2.25) {};
		\node [style=left label] (24) at (1.25, -1.5) {$B_{aux}$};
		\node [style=none] (25) at (1.25, -0.75) {};
		\node [style=none] (26) at (1.25, 0.75) {};
		\node [style=none] (28) at (1.25, 0.25) {};
		\node [style=upground] (29) at (1.25, 1) {};
		\node [style=none] (30) at (4.25, -2.25) {};
		\node [style=right label] (31) at (4.25, -1.5) {$B'$};
		\node [style=none] (32) at (4.25, 1.5) {};
		\node [style=none] (33) at (1.75, 1.5) {};
		\node [style=none] (34) at (4.75, 1.5) {};
		\node [style=none] (35) at (3.25, 2.75) {};
		\node [style=none] (36) at (3.25, 2) {$\Phi_b$};
	\end{pgfonlayer}
	\begin{pgfonlayer}{edgelayer}
		\draw (19.center) to (16.center);
		\draw (16.center) to (17.center);
		\draw (17.center) to (18.center);
		\draw (18.center) to (19.center);
		\draw [qWire] (21.center) to (8.center);
		\draw [qWire] (22.center) to (7.center);
		\draw [qWire] (25.center) to (23.center);
		\draw [qWire] (28.center) to (26.center);
		\draw [qWire] (32.center) to (30.center);
		\draw (35.center) to (33.center);
		\draw (33.center) to (34.center);
		\draw (34.center) to (35.center);
	\end{pgfonlayer}
\end{tikzpicture}
}\,.
\end{align*}
\blk The idea is that these dichotomic measurements will allow us to construct the matrix elements of $\Gamma_n$ that correspond to rows and columns whose labels contain the letter $y$ (see Eq.~\eqref{eq:theOy} below). For this to be the case, it is crucial to notice \blk that the measurement $\{\Psi_{1|y}, \Psi_{2|y} \}$, for each $y \in \Y$, is a projective measurement. Indeed, in conventional notation the corresponding operator may be written as:
\begin{align}\label{eq:bobmeas}
\Psi_{b|y} = \left(U^\dagger_y \otimes \id_{B'}\right) \left( \id_{B_{aux}} \otimes  \Phi_b \right) \left(U_y \otimes \id_{B'}\right)\,. 
\end{align}
Therefore, the following are satisfied: \blk
\begin{compactitem}
\item[(i)] Normalisation: $\{\Psi_{1|y}, \Psi_{2|y} \}$ is a complete measurement for each $y \in \Y$, since: 
\begin{align*}
\sum_{b=1:2} \Psi_{b|y} &=  \left(U^\dagger_y \otimes \id_{B'}\right) \left( \id_{B_{aux}} \otimes  \id_{BB'} \right) \left(U_y \otimes \id_{B'}\right) = \id_{B_{aux}B_{(i)}B'}\,.
\end{align*}
\item[(ii)] Hermiticity: $\Psi_{b|y}$ is effectively a unitary applied to a state, and since unitaries preserve hermiticity then $\Psi_{b|y}$ is Hermitian,  
\item[(iii)] Idempotency:  $\Psi_{b|y}^\dagger\Psi_{b|y} = \Psi_{b|y}$ since: 
\begin{align*}
\Psi_{b|y}^\dagger\Psi_{b|y} &= \left(U^\dagger_y \otimes \id_{B'}\right) \left( \id_{B_{aux}} \otimes  \Phi_b \right) \left(U_y \otimes \id_{B'}\right) \left(U^\dagger_y \otimes \id_{B'}\right) \left( \id_{B_{aux}} \otimes  \Phi_b \right) \left(U_y \otimes \id_{B'}\right) \\
&= \left(U^\dagger_y \otimes \id_{B'}\right) \left( \id_{B_{aux}} \otimes  \Phi_b \right) \left( \id_{B_{aux}} \otimes  \Phi_b \right) \left(U_y \otimes \id_{B'}\right) \\
&= \left(U^\dagger_y \otimes \id_{B'}\right) \left( \id_{B_{aux}} \otimes  \Phi_b \right) \left(U_y \otimes \id_{B'}\right) \\
&= \Psi_{b|y}\,.
\end{align*}
\end{compactitem}
Now we have all the ingredients to prove the main theorem of this section:
\begin{thm}\label{thm:quantumInQn}\textbf{Quantum assemblages are $\cQ^n$ assemblages $\forall \, n \in \mathbb{N}$.}\\
Let $\As_\mathrm{Q} = \{\sigma_{a|xy}\}_{a \in \A, \, x \in \X, \, y \in \Y}$ be a quantum assemblage in the Bob-with-Input EPR scenario. Then $\As_\mathrm{Q} \in \cQ^n$ $\forall \, n \in \mathbb{N}$.
\end{thm}
\begin{proof}
The proof strategy is to construct a moment matrix $\Gamma_n$ that serves as a certificate of order $n$ for the assemblage $\As_\mathrm{Q}$.
For this, let us first define a set of operators $\{\Op_\mathbf{v}\}_{\mathbf{v} \in \cS_n}$, labelled by the words of the set $\cS_n$, which act on the Hilbert space $\cH=\cH_A \otimes \cH_{B_{aux}} \otimes \cH_{B_{(i)}}\otimes \cH_{B'}$. For the case of single-letter words (i.e., $\mathbf{v} \in \Upsilon$), the operators are defined as:
\begin{align}
\Op_\emptyset &:= \id_{AB_{aux}B_{(i)}B'} \,,\\
\Op_{a|x} &:= \Pi_{a|x} \otimes \id_{B_{aux}B_{(i)}B'} \quad \forall \, a|x \in \Upsilon \,,\\
\Op_{y} &:= \id_A \otimes \Psi_{1|y} \quad \forall \, y \in \Upsilon\,. \label{eq:theOy}
\end{align}
For the case of works with multiple letters, i.e., $\mathbf{v} = \mathbf{v}_1 \ldots \mathbf{v}_k$ with $1 < k \leq n$ and $\mathbf{v}_j \in \Upsilon \quad \forall 1 \leq j \leq k$, the operators are defined as:
\begin{align}
\Op_{\mathbf{v}} := \Op_{\mathbf{v}_1} \ldots \Op_{\mathbf{v}_k} \,.
\end{align}
Hence, 
\begin{align}
\Op_{\mathbf{v}} \Op_{\mathbf{w}} = \Op_{\mathbf{vw}}.
\end{align}
Notice also that:
\begin{align}
\Op_\emptyset \Op_{\mathbf{v}} &= \Op_{\mathbf{v}} = \Op_{\mathbf{v}} \Op_{\emptyset} \quad \forall \, \mathbf{v} \in \cS_{n-1} \,,\\
\Op_{a|x} \Op_y &= \Op_y \Op_{a|x} \quad \forall \, a|x , y \in \Upsilon \,\\
\Op_{\mathbf{v}} \Op_{\mathbf{v}} &= \Op_{\mathbf{v}} \quad \forall \, \mathbf{v} \in \Upsilon\,,
\end{align} 
therefore 
\begin{align}
\mathbf{v} \equiv \mathbf{w} \, \Rightarrow \, \Op_{\mathbf{v}} = \Op_{\mathbf{w}}.
\end{align}
In addition, $\Op_{a|x} \Op_{a'|x} = \mathbf{0}_{\cH}$ (i.e., the zero matrix in the Hilbert space $\cH$) whenever $a\neq a'$ for all $x \in \X$. Therefore
\begin{align}
\Op_{\mathbf{v}} =  \mathbf{0}_{\cH} \quad \text{when} \,\, \mathbf{v} \,\, \text{is null} \,.
\end{align}
Finally, since the operators for single-lettered words are Hermitian, it follows that:
\begin{align}
\Op_\mathbf{v}^\dagger = \Op_{\mathbf{v}^\dagger}\,.
\end{align}

The last object that we need to define to set up the ingredients for the proof, is a state $\tilde{\rho}$ for the joint system $A \cdot B_{aux} \cdot B_{(i)} \cdot B'$. 
Following the notation of this section, let $\ket{\psi}\bra{\psi} := \rho'$ be the pure quantum state that realises the assemblage $\As_{\mathrm{Q}}$ as per Eq.~\eqref{eq:therhop}. Then, define  the state $\tilde{\rho}$ as
\begin{align}
\tilde{\rho} := \ket{\psi}\bra{\psi} \otimes \id_{B'}\,.
\end{align}
With these, define a square matrix $\Gamma_n$ of dimension $|\cS_n|$, whose rows and columns are indexed by the words in $\cS_n$, as:
\begin{align}
\Gamma_n(\mathbf{v},\mathbf{w}) := \tr{\Op_{\mathbf{v}}^\dagger \, \Op_{\mathbf{w}} \, \tilde{\rho}}{AB_{aux}B_{(i)}}\,.
\end{align}
Notice that the elements of $\Gamma_n$ are $d \times d$ complex matrices, where $d$ is the dimension of $\cH_{B'}$. Next we will show that $\Gamma_n$ is a certificate of order $n$ for the assemblage $\As_\mathrm{Q}$.

Let us first begin by showing that $\Gamma_n$ is indeed a moment matrix of order $n$. For this, we need to show that its element satisfy the five sets of conditions listed in Definition \ref{def:MMn} via Eqs.~\eqref{eq:2} to \eqref{eq:6}:
\begin{compactitem}
\item The $(i,j)$ entry of the element $\Gamma_n(\mathbf{v},\mathbf{w})$ is 
\begin{align*}
\Gamma_n^{i,j}(\mathbf{v},\mathbf{w}) &= \bra{i} \tr{\Op_{\mathbf{v}}^\dagger \, \Op_{\mathbf{w}} \, \tilde{\rho}}{AB_{aux}B_{(i)}} \ket{j} \\
&=\bra{i} \tr{\Op_{\mathbf{v}}^\dagger \, \Op_{\mathbf{w}} \, \ket{\psi} \bra{\psi} \otimes \id_d^{B'} }{AB_{aux}B_{(i)}} \ket{j} \\
&= \bra{i} \bra{\psi} \Op_{\mathbf{v}}^\dagger \, \Op_{\mathbf{w}} \ket{\psi} \ket{j}\,.
\end{align*}
Now define the vectors $\ket{\psi_{i,\mathbf{v}}} = \Op_{\mathbf{v}} \ket{\psi} \ket{i}$. Hence, $\Gamma_n^{i,j}(\mathbf{v},\mathbf{w}) = \langle \psi_{i,\mathbf{v}} \lvert \psi_{i,\mathbf{w}} \rangle$. This shows that $\Gamma$ is a Gramian matrix, and therefore positive semidefinite. Hence, Eq.~\eqref{eq:2} is satisfied.
\item $\Gamma_n(\emptyset,\emptyset) = \tr{\id_\cH \, \id_\cH \, \tilde{\rho}}{AB_{aux}B_{(i)}} =  \id_{B'}$, hence Eq.~\eqref{eq:3} is satisfied.
\item On the one hand, $\Gamma_n(\mathbf{v},\mathbf{w}) = \tr{\Op_{\mathbf{v}}^\dagger \, \Op_{\mathbf{w}} \, \tilde{\rho}}{AB_{aux}B_{(i)}} = \tr{\Op_{\mathbf{v}^\dagger\mathbf{w}} \, \tilde{\rho}}{AB_{aux}B_{(i)}}$. \\
On the other hand, $\Gamma_n(\mathbf{v}',\mathbf{w}') = \tr{\Op_{\mathbf{v}'}^\dagger \, \Op_{\mathbf{w}'} \, \tilde{\rho}}{AB_{aux}B_{(i)}} = \tr{\Op_{\mathbf{v}'^\dagger\mathbf{w}'} \, \tilde{\rho}}{AB_{aux}B_{(i)}}$. \\ If $\mathbf{v}^\dagger\mathbf{w} \equiv \mathbf{v}'^\dagger\mathbf{w}'$, then $\Op_{\mathbf{v}^\dagger\mathbf{w}} = \Op_{\mathbf{v}'^\dagger\mathbf{w}'}$, and hence $\Gamma_n(\mathbf{v},\mathbf{w}) = \Gamma_n(\mathbf{v}',\mathbf{w}')$. It follows that Eq.~\eqref{eq:4} is satisfied.
\item If $\mathbf{v}^\dagger\mathbf{w}$ is null, then $\Op_{\mathbf{v}^\dagger\mathbf{w}} =  \mathbf{0}_{\cH} $. Hence, $\Gamma_n(\mathbf{v},\mathbf{w}) = \tr{\Op_{\mathbf{v}^\dagger\mathbf{w}} \, \tilde{\rho}}{AB_{aux}B_{(i)}} = \mathbf{0}_d$, and so Eq.~\eqref{eq:5} is satisfied. 
\item For any $a|x, a'|x' \in \Upsilon$, 
\begin{align*}
\Gamma_{a|x, a'|x'} &= \tr{ (\Pi_{a|x} \Pi_{a'|x'} \otimes \id_{B_{aux}B_{(i)}B'}  ) \, \tilde{\rho}}{AB_{aux}B_{(i)}} \\
&= \tr{ \Pi_{a|x} \Pi_{a'|x'} \otimes \id_{B_{aux}B_{(i)}}   \, \ket{\psi}\bra{\psi} }{} \, \id_{B'} \,.
\end{align*}
Hence, $\Gamma_{a|x, a'|x'} \propto \id_d$ and Eq.~\eqref{eq:6} follows. 
\end{compactitem}
Now that we've shown that $\Gamma_n$ is indeed a moment matrix of order $n$, we need to show that it is moreover a certificate of order $n$ for the assemblage $\As_{\mathrm{Q}}$. For this, we need to now show that $\Gamma_n$ satisfies the conditions given by Eqs.~\eqref{eq:7} to \eqref{eq:9} in Definition \ref{def:cert}. 
\begin{compactitem}
\item For any $a|x \in \Upsilon$,
\begin{align*}
\Gamma_n(\emptyset,a|x) &= \tr{\id_\cH \, \Pi_{a|x} \otimes \id_{B_{aux}B_{(i)}B'} \, \tilde{\rho}}{AB_{aux}B_{(i)}} \\
&= \tr{\Pi_{a|x} \otimes \id_{B_{aux}B_{(i)}B'} \, \ket{\psi}\bra{\psi} \otimes \id_{B'}}{AB_{aux}B_{(i)}} \\
&= \tr{\Pi_{a|x} \otimes \id_{B_{aux}B_{(i)}} \, \ket{\psi}\bra{\psi}}{} \,  \id_{B'} \\
&= \tr{\Pi_{a|x} \otimes \id_{B_{aux}B_{(i)}} \, (\id_A \otimes U_y) \ket{\psi}\bra{\psi} (\id_A \otimes U_y^\dagger) }{} \,  \id_{B'}\,,
\end{align*}
since unitary operations are trace preserving. Hence,
\begin{align*}
\Gamma_n(\emptyset,a|x) = \tr{\sigma_{a|xy}}{B} \,  \id_{B'}  = \tr{\sigma_{a|x1}}{} \,  \id_{d}  \,,
\end{align*}
where the last equality follows from the fact that the assemblage $\As_\mathrm{Q}$ is no-signalling from Bob to Alice, by definition. Hence, Eq.~\eqref{eq:7} is satisfied. 
\item For any $a|x, y \in \Upsilon$,
\begin{align*}
\Gamma_n(\emptyset,a|x y) &= \tr{\id_\cH \, \Pi_{a|x} \otimes \Psi_{1|y} \, \tilde{\rho}}{AB_{aux}B_{(i)}} \\
&= \tr{\Pi_{a|x} \otimes \left( \left(U^\dagger_y \otimes \id_{B'}\right) \left( \id_{B_{aux}} \otimes  \Phi_1 \right) \left(U_y \otimes \id_{B'}\right) \right) \, \tilde{\rho}}{AB_{aux}B_{(i)}} \\
&= \tr{\left( \Pi_{a|x}^\dagger \otimes  U^\dagger_y \otimes \id_{B'} \right)   \left( \id_{AB_{aux}} \otimes  \Phi_1 \right) \left( \Pi_{a|x} \otimes  U_y \otimes \id_{B'} \right)    \, \tilde{\rho}    }{AB_{aux}B_{(i)}} \\
&= \tr{   \left( \id_{AB_{aux}} \otimes  \Phi_1 \right)   \left( \Pi_{a|x} \otimes  U_y \otimes \id_{B'} \right)    \, \tilde{\rho}  \,  \left( \Pi_{a|x}^\dagger \otimes  U^\dagger_y \otimes \id_{B'} \right)   }{AB_{aux}B_{(i)}} \\
&= \tr{   \left( \id_{AB_{aux}} \otimes  \Phi_1 \right)   \left( \left( \Pi_{a|x} \otimes  U_y \right)   \ket{\psi}\bra{\psi}  \left( \Pi_{a|x}^\dagger \otimes  U^\dagger_y  \right) \otimes \id_{B'} \right)   }{AB_{aux}B_{(i)}} \\
&= \tr{   \Phi_1   \left( \sigma_{a|xy} \otimes \id_{B'} \right)   }{B} \\
&= \frac{1}{d} \, \sigma_{a|xy}^{\mathrm{T}}\,.
\end{align*}
Hence, Eq.~\eqref{eq:9} follows. 
\item This last condition follows similarly from the previous one plus the fact that Alice's measurements $\{\Pi_{a|x}\}$ are complete:
\begin{align*}
\Gamma_n(\emptyset,y) &= \tr{\id_\cH \, \id_A \otimes \Psi_{1|y} \, \tilde{\rho}}{AB_{aux}B_{(i)}} \\
&= \sum_{a \in \A}  \tr{\id_\cH \, \Pi_{a|x} \otimes \Psi_{1|y} \, \tilde{\rho}}{AB_{aux}B_{(i)}} \\
&= \sum_{a \in \A} \frac{1}{d} \, \sigma_{a|xy}^{\mathrm{T}}\\
&= \frac{1}{d} \, \sigma_{y}^{\mathrm{T}}\,,
\end{align*}
It hence follows that the matrix $\Gamma_n$ satisfies the conditions of Eq.~\eqref{eq:8}. 
\end{compactitem}

We have hence shown that the matrix $\Gamma_n$ that we constructed is a certificate of order $n$ for the assemblage $\As_{\mathrm{Q}}$, which implies that $\As_{\mathrm{Q}} \in \cQ^n$, for all $n \in \mathbb{N}$.
\end{proof}

\subsection{Convergence}\label{se:convergence}

Let $\cQ$ denote the set of all quantumly-realisable assemblage in a Bob-with-Input scenario. In the previous section we showed that $\cQ \subseteq \cQ^n \,\, \forall \, n \in \mathbb{N}$. Here, we will show that the hierarchy $\{\cQ^n\}_{ n \in \mathbb{N}}$ converges to a set of assemblages $\cQ_C$ that therefore contains the quantum set. We leave it an open question to check under which conditions $\cQ_C \equiv \cQ$. \blk The proof is rather technical, and leverages algebraic techniques similar to those used in related hierarchies of semidefinite tests \cite{NPA,Navascu_s_2008,acin2015combinatorial}. \blk

We start with some preliminary algebraic constructions and notations. Let $(\mathcal{H}_B,\braket{.|.}_B)$ be the $d$-dimensional Hilbert space for Bob. First, we define a sequence of Hilbert spaces $(\mathcal{H}_m)_{m\in\mathbb{N}}$ such that for each $m$ we can regard a moment matrix of order $m$ as a bounded linear operator on $\mathcal{H}_m$. Let
\begin{equation}
    \mathcal{M}(\mathcal{S},\mathcal{H}_B)
\end{equation}
be the vector space of all maps from the set of words $\mathcal{S}$ to Bob's Hilbert space $\mathcal{H}_B$ with addition and scalar multiplication defined pointwise, i.e.\ adding two such maps means adding the values for each argument. For a word $v\in\mathcal{S}$ and a vector $x\in\mathcal{H}_B$ we define the element
\begin{equation}
    (x\cdot[v])\in\mathcal{M}(\mathcal{S},\mathcal{H}_B)
\end{equation}
as follows: For $w\in\mathcal{S}$ let
\begin{equation}
    (x\cdot[v])(w):=\begin{cases}
        \begin{array}{cl}
            x & :\text{ if $w=v$}\\
            0 & :\text{ if $w\neq v$}
        \end{array}
    \end{cases}
\end{equation}
Now we define $\mathcal{H}_m$ as
\begin{equation}
    \mathcal{H}_m:=\operatorname{span}\set{(x\cdot[v])| x\in\mathcal{H}_B, v\in\mathcal{S}_m}\subseteq \mathcal{M}(\mathcal{S},\mathcal{H}_B),
\end{equation}
i.e.\ $\mathcal{H}_m$ is the set of all maps from $\mathcal{S}$ to $\mathcal{H}_B$ that are $0$ outside of $\mathcal{S}_m$. The inner product on $\mathcal{H}_m$ is defined as follows: For $f,g\in\mathcal{H}_m$ let
\begin{equation}
    \braket{f|g}:=\sum_{v\in\mathcal{S}_m}\braket{f(v)|g(v)}_B.
\end{equation}
Note that the sum on the right is finite.

Now let $\Sigma$ be an assemblage such that $\Sigma\in\cQ^n$ for all $n$. Then for each $n$ there is a certificate $\Gamma_n$ of order $n$ for $\Sigma$. For $v,w\in\mathcal{S}\backslash\mathcal{S}_n$ let $\Gamma_n(v,w):=0$. Then we can regard each $\Gamma_n$ as a bounded linear operator $\Gamma_n^{(m)}\in\mathcal{B}(\mathcal{H}_m)$ as follows: For $f\in\mathcal{H}_m$ let
\begin{equation}
    \label{eq:gammaoperator}
    \Gamma_n^{(m)}(f):=\sum_{v,w\in\mathcal{S}_m} ((\Gamma_n(v,w)f(w))\cdot[v]).
\end{equation}
Note that we will use the same definition for limit operators $\Gamma_{\infty}^{(m)}$ later on.

\begin{lem}
    \label{lem:diagonal<=id}
    For all natural numbers $n$ and for all words $v\in\mathcal{S}$ we have
    \begin{equation}
        \Gamma_n(v,v)\leq\operatorname{id}_{\mathcal{H}_B}.
    \end{equation}
\end{lem}
\begin{proof}
    We prove the claim by induction on the length $l$ of $v$. For $l=0$ we have $v=\emptyset$ and $\Gamma(\emptyset,\emptyset)=\operatorname{id}_{\mathcal{H}_B}$ by the definition of a moment matrix. Now let $l>0$. Then there is a word $w$ of length $l-1$ and a letter $c\in\Upsilon$ such that
    \begin{equation}
        v=cw.
    \end{equation}
    Let
    \begin{equation}
        \begin{split}
            A &:= \Gamma_n(w,w)\\
            B &:= \Gamma_n(v,v).
        \end{split}
    \end{equation}
    Then we have
    \begin{equation}
        \label{eq:A<=id}
        A\leq\operatorname{id}_{\mathcal{H}_B}
    \end{equation}
    by the induction hypothesis and
    \begin{equation}
        \Gamma_n(w,v)=\Gamma_n(v,w)=\Gamma_n(v,v)=B
    \end{equation}
    because
    \begin{equation}
        v^{\dag}v=w^{\dag}ccw\equiv w^{\dag}cw=w^{\dag}v.
    \end{equation}
    It follows that
    \begin{equation}
        0\leq\begin{pmatrix}
            A & B\\
            B & B
        \end{pmatrix}
    \end{equation}
    since the matrix on the right is a contraction of the positive semidefinite matrix $\Gamma_n$. Now this is equivalent to
    \begin{equation}
        0\leq B\leq A\stackrel{\eqref{eq:A<=id}}{\leq} \operatorname{id}_{\mathcal{H}_B}
    \end{equation}
    and thus
    \begin{equation}
        \Gamma_n(v,v)=B\leq\operatorname{id}_{\mathcal{H}_B}
    \end{equation}
    as desired.
\end{proof}

\begin{lem}
    \label{lem:normbounded}
    For all natural numbers $n$ and for all words $v,w\in\mathcal{S}$ we have
    \begin{equation}
        \lVert \Gamma_n(v,w)\rVert\leq 1.
    \end{equation}
\end{lem}
\begin{proof}
    For $v=w$ this follows immediately from Lemma~\ref{lem:diagonal<=id}. For $v\neq w$ we consider the contraction of $\Gamma_n$ for the rows and columns corresponding to $v$ and $w$. Then we have
    \begin{equation}
        0\leq\begin{pmatrix}
            \Gamma_n(v,v) & \Gamma_n(v,w)\\
            \Gamma_n(w,v) & \Gamma_n(w,w)
        \end{pmatrix}\leq\begin{pmatrix}
            \operatorname{id}_{\mathcal{H}_B} & \Gamma_n(v,w)\\
            \Gamma_n(w,v) & \operatorname{id}_{\mathcal{H}_B},
        \end{pmatrix}
    \end{equation}
    where the first inequality follows from the positivity of $\Gamma_n$ and the second from Lemma~\ref{lem:diagonal<=id}. This implies
    \begin{equation}
        \Gamma_n(v,w)^{\dag}\Gamma_n(v,w)=\Gamma_n(w,v)\Gamma_n(v,w)\leq\operatorname{id}_{\mathcal{H}_B}.
    \end{equation}
    This in turn implies
    \begin{equation}
        \lVert\Gamma_n(v,w)\rVert^2=\lVert\Gamma_n(v,w)^{\dag}\Gamma_n(v,w)\rVert\leq 1
    \end{equation}
    from which the claim follows.
\end{proof}

\begin{thm}
    \label{thm:limitoperator}
    There is a subsequence of indices $(n_k)_{k\in\mathbb{N}}$ and an infinite matrix
    \begin{equation}
        \Gamma_{\infty}:=(\Gamma_{\infty}(v,w))_{v,w\in\mathcal{S}}
    \end{equation}
    whose entries are in $\mathcal{B}(\mathcal{H}_B)$, and whose rows and columns are indexed by the words in $\mathcal{S}$ such that for every natural number $m$ the sequence $(\Gamma_{n_k}^{(m)})_{k\in\mathbb{N}}$ converges in the operator norm to an operator $\Gamma_{\infty}^{(m)}\in\mathcal{B}(\mathcal{H}_B)$ that is defined as in Eq.~\eqref{eq:gammaoperator} for $n=\infty$ and is a certificate of order $m$ for $\Sigma$.
\end{thm}
\begin{proof}
    First, we choose an arbitrary basis of the Hilbert space $\mathcal{H}_B$ and an enumeration of the countable set $\mathcal{S}^d\times\mathcal{S}^d$. Then by Lemma~\ref{lem:normbounded} we may regard each moment matrix $\Gamma_n=(\Gamma_n(v,w))_{v,w\in\mathcal{S}}$ as an element of the Banach space
    \begin{equation}
        \ell^{\infty}:=\ell^{\infty}(\mathbb{N}),
    \end{equation}
    the space of all bounded complex sequences, equipped with the supremum norm. Note that $\ell^1(\mathbb{N})$ is a continuous pre-dual space of $\ell^{\infty}$. By the Banach-Alaoglu-Theorem the closed unit ball of $\ell^{\infty}$ is therefore compact with respect to the weak-* topology. Thus, there is a subsequence of indices $(n_k)_{k\in\mathbb{N}}$ such that $\Gamma_{n_k}$ converges in the weak-* topology to an element in $\ell^{\infty}$ that we can identify with a matrix $\Gamma_{\infty}:=(\Gamma_{\infty}(v,w))_{v,w\in\mathcal{S}}$. Now define the operators $\Gamma_{\infty}^{(m)}$ as in Eq.~\eqref{eq:gammaoperator}. Note that for each fixed $m$ only finitely many words are actually involved in each operator. Since convergence in the weak-* topology in particular implies pointwise convergence, we get that the sequence $(\Gamma_{n_k}^{(m)})_{k\in\mathbb{N}}$ converges to $\Gamma_{\infty}^{(m)}$ in the operator norm of $\mathcal{B}(\mathcal{H}_m)$. Because all constraints on a certificate behave continuously, $\Gamma_{\infty}^{(m)}$ must be a certificate of order $m$ for $\Sigma$.
\end{proof}

Theorem~\ref{thm:limitoperator} shows that our hierarchy $\{\cQ^n\}_{n\in\mathbb{N}}$ converges to a set of assemblages $\cQ_C$ that contains the set of quantum assemblages by Theorem~\ref{thm:quantumInQn}. 

In order to show that every assemblage $\Sigma=\{\sigma_{a|xy}\}_{a \in \A, \, x \in \X, \, y \in \Y}\in\cQ_C$ is quantum -- a question that we leave open -- we have to find a Hilbert space $\mathcal{K}$, a shared state $\rho\in\mathcal{B}(\mathcal{K})$, measurements for Alice $\Pi_{a|x}$, CPTP maps $\cE^y\colon\mathbb{B}(\mathcal{H}_B)\to\mathcal{B}(\mathcal{H}_B)$ for Bob and a CPTP map $\tau\colon\mathcal{B}(\mathcal{K})\to\mathcal{B}(\mathcal{H}_B)$ such that
\begin{equation}\label{eq:theneed}
    \sigma_{a|xy}=\cE^y(\tau(\Pi_{a|x}\rho))
\end{equation}
for all $a \in \A, \, x \in \X, \, y \in \Y$.

In the remainder of this section we will construct some promising candidates for all the ingredients in the right-hand-side of Eq.~\eqref{eq:theneed} but the map $\tau$, based on the limit $\Gamma_{\infty}=(\Gamma_{\infty}(v,w))_{v,w\in\mathcal{S}}$ established in Theorem~\ref{thm:limitoperator}. Similar to the construction of the Hilbert spaces $\mathcal{H}_m$ above, we construct a $*$-algebra $\mathcal{A}$ from the set of words $\mathcal{S}$. Let
\begin{equation}
    \mathcal{M}(\mathcal{S},\C)
\end{equation}
be the vector space of all maps from the set of words $\mathcal{S}$ to the complex numbers with addition and scalar multiplication defined pointwise. Now let
\begin{equation}
    \mathcal{A}:=\{f\in\mathcal{M}(\mathcal{S},\C)| \operatorname{supp}(f)\text{ is finite}\}
\end{equation}
be the subspace of all maps that vanish for all except finitely many words. For $v\in\mathcal{S}$ we define the element
\begin{equation}
    [v]\in\mathcal{A}
\end{equation}
as follows: For $w\in\mathcal{S}$ let
\begin{equation}
    [v](w):=\begin{cases}
        \begin{array}{cl}
            1 & :\text{ if $w=v$}\\
            0 & :\text{ if $w\neq v$}
        \end{array}
    \end{cases}
\end{equation}
We further equip $\mathcal{A}$ with a multiplication and an involution as follows: For $f,g\in\mathcal{A}$ and $v\in\mathcal{S}$ let
\begin{equation}
    (f\cdot g)(v):=\sum_{\substack{{u,w\in\mathcal{S}}\\{uw=v}}} f(u)g(w)
\end{equation}
and
\begin{equation}
    f^{\dag}(v):=f(v^{\dag})^*.
\end{equation}
This makes $\mathcal{A}$ a unital $*$-algebra with identity element $[\emptyset]$. Finally, let $\mathcal{I}$ be the subspace of $\mathcal{A}$ spanned by all elements of the form
\begin{equation}
    [uvvw]-[uvw],\quad\quad [u\,a|x\,y\, w]-[u\,y\,a|x\,w],\quad\quad [u\,a|x\,a'|x\,w]
\end{equation}
for $u,w\in\mathcal{S}$, $v\in\Upsilon$, $a,a'\in\mathbb{A}$ with $a\neq a'$ and $y\in\mathbb{Y}$. Then $\mathcal{I}$ is a twosided $*$-ideal of $\mathcal{A}$ that represents the equivalence relation between words.

Now we are ready to define a $*$-linear map on our $*$-algebra $\mathcal{A}$. Let $\Gamma=\Gamma_{\infty}$ be as in Theorem~\ref{thm:limitoperator} and define
\begin{equation}
    \varphi\colon\mathcal{A}\to\mathcal{B}(\mathcal{H}_B),\quad f\mapsto \sum_{v\in\mathcal{S}} f(v)\cdot\Gamma(\emptyset,v).
\end{equation}
Note that this is well-defined because each $f\in\mathcal{A}$ has finite support by definition. The linearity of this map is obvious. For the compatibility with the involution we compute
\begin{equation}
    \begin{split}
        \varphi(f^{\dag}) &= \sum_{v\in\mathcal{S}} f^{\dag}(v)\cdot\Gamma(\emptyset,v) = \sum_{v\in\mathcal{S}} f(v^{\dag})^*\cdot\Gamma(\emptyset,v)\\
        &= \sum_{v\in\mathcal{S}} f(v^{\dag})^*\cdot\Gamma(v^{\dag},\emptyset) = \sum_{v\in\mathcal{S}} f(v^{\dag})^*\cdot\Gamma(\emptyset,v^{\dag})^{\dag}\\
        &= \varphi(f)^{\dag},
    \end{split}
\end{equation}
where the third equality follows from $(v^{\dag})^{\dag}\emptyset\equiv\emptyset^{\dag}v$.

\begin{prop}
    \label{prop:completelypositive}
    The $*$-linear map $\varphi$ is completely positive and $\mathcal{I}\subseteq\operatorname{ker}(\varphi)$.
\end{prop}
\begin{proof}
    Let $n$ be a natural number and let $F=(f_{i,j})_{i,j}\in\operatorname{Mat}_n(\mathcal{A})$. We have to show that
    \begin{equation}
        (\operatorname{id}_n\otimes\varphi)(F^{\dag}F)\in\operatorname{Mat}_n(\mathcal{B}(\mathcal{H}_B))
    \end{equation}
    is positive. Here $\operatorname{id}_n$ denotes the identity map on $\operatorname{Mat}_n(\mathbb{C})$. We choose a vector $x=(x_1,\dots,x_n)^t\in\mathcal{H}_B^n$. Now we compute
    \begin{equation}
        \begin{split}
            \braket{ x|(\operatorname{id}_n\otimes\varphi)(F^{\dag}F)| x} &= \sum_{i,j=1}^{n}\braket{ x_i|\varphi((F^{\dag}F)_{i,j})| x_j} = \sum_{i,j,k=1}^{n} \braket{ x_i|\varphi(f_{k,i}^{\dag}f_{k,j})| x_j}\\
            &= \sum_{i,j,k=1}^{n}\sum_{v\in\mathcal{S}} \braket{ x_i|(f_{k,i}^{\dag}\cdot f_{j,k})(v)\cdot\Gamma(\emptyset,v)| x_j}\\
            &= \sum_{i,j,k=1}^{n}\sum_{v,w\in\mathcal{S}} \braket{ x_i| f_{k,i}^{\dag}(v)\cdot f_{k,j}(w)\cdot\Gamma(\emptyset,vw)| x_j}\\
            &= \sum_{i,j,k=1}^{n}\sum_{v,w\in\mathcal{S}} \braket{ x_i| f_{k,i}(v^{\dag})^*\cdot f_{k,j}(w)\cdot\Gamma(v^{\dag},w)| x_j}\\
            &= \sum_{i,j,k=1}^{n}\sum_{v,w\in\mathcal{S}} \braket{ x_i| f_{k,i}(v)^*\cdot f_{k,j}(w)\cdot\Gamma(v,w)| x_j}.
        \end{split}
    \end{equation}
    The second to last equality follows from the fact that $(v^{\dag})^{\dag}w\equiv\emptyset^{\dag}vw$. Now consider the set
    \begin{equation}
        \mathcal{W}=\bigcup_{i,j=1}^{n}\operatorname{supp}(f_{i,j}).
    \end{equation}
    Then $\mathcal{W}$ is a finite set. Thus there is a natural number $m$ such that $\mathcal{W}\subseteq\mathcal{S}_m$. For $k=1,\dots,n$ we define
    \begin{equation}
        y_k:=\sum_{v\in\mathcal{S}_m} \left(\sum_{i=1}^{n} f_{k,i}(v)\cdot x_i\right)\cdot[v]\in\mathcal{H}_m.
    \end{equation}
    Resuming our computation from above we have
    \begin{equation}
        \begin{split}
            \braket{ x|(\operatorname{id}_n\otimes\varphi)(F^{\dag}F)| x} &= \sum_{i,j,k=1}^{n}\sum_{v,w\in\mathcal{S}} \braket{ x_i| f_{k,i}(v)^*\cdot f_{k,j}(w)\cdot\Gamma(v,w)| x_j}\\
            &= \sum_{k=1}^{n} \sum_{v,w\in\mathcal{S}_m} \braket{ y_k(v)|\Gamma(v,w)| y_k(w)}\\
            &= \sum_{k=1}^{n} \braket{ y_k| \Gamma_{\infty}^{(m)}| y_k}\geq 0.
        \end{split}
    \end{equation}
    The inequality follows from the fact that $\Gamma_{\infty}^{(m)}$ is a certificate of order $m$ and thus a positive operator by Theorem~\ref{thm:limitoperator}. It remains to show that $\mathcal{I}\subseteq\operatorname{ker}(\varphi)$. Since $\varphi$ is a linear map, we only have to show that the generators of $\mathcal{I}$ are elements of $\operatorname{ker}(\varphi)$. However, this easily follows from the definition of a moment matrix.
\end{proof}

By the Homomorphism Theorem and since $\mathcal{I}\subseteq\operatorname{ker}(\varphi)$, the map
\begin{equation}\label{eq:62}
    \psi\colon \mathcal{A}/\mathcal{I}\to\mathcal{B}(\mathcal{H}_B),\quad x+\mathcal{I}\mapsto\varphi(x)
\end{equation}
is well-defined, $*$-linear, and completely positive. It is well known that $\psi$ is completely positive if and only if the map
\begin{equation}
    s_{\psi}\colon(\mathcal{A}/\mathcal{I})\otimes\mathcal{B}(\mathcal{H}_B)\to\C,\quad (x+\mathcal{I})\otimes T\mapsto\Tr{\psi(a+\mathcal{I})T^t}
\end{equation}
is a positive $*$-linear functional. By the GNS-construction there is a Hilbert space $\mathcal{K}$, a cyclic vector $\xi\in\mathcal{K}$, and a $*$-algebra homomorphism
\begin{equation}
    \pi\colon(\mathcal{A}/\mathcal{I})\otimes\mathcal{B}(\mathcal{H}_B)\to\mathcal{B}(\mathcal{K})
\end{equation}
such that for all $x\in(\mathcal{A}/\mathcal{I})\otimes\mathcal{B}(\mathcal{H}_B)$ we have
\begin{equation}
    s_{\psi}(x)=\braket{\xi|\pi(x)|\xi}.
\end{equation}

Now let $\mathcal{K}$ be the shared Hilbert space,
\begin{equation}
    \rho:=\ket{\xi}\bra{\xi}\in\mathcal{B}(\mathcal{K}),
\end{equation}
and
\begin{equation}
    \Pi_{a|x}:=\pi(((a|x)+\mathcal{I})\otimes\operatorname{id}_{\mathcal{H}_B}).
\end{equation}
Once we can find an appropriate CPTP map $\tau\colon\mathcal{B}(\mathcal{K})\to\mathcal{B}(\mathcal{H}_B)$ we can further define
\begin{equation}\label{eq:67}
    \cE^y(T):=\tau(\pi((y+\mathcal{I})\otimes T)).
\end{equation}

The open question then is how to define $\tau$. Take for example the case where we start from the quantum data of Theorem~\ref{thm:quantumInQn} and, as constructed there, obtain its representation in terms of a moment matrix of each order. Then, one can take such moment matrices and reverse-engineer the quantum data by applying the construction outlined between Eqs.~\eqref{eq:62} to \eqref{eq:67} and taking the map $\tau$ to be `tracing out $\mathcal{H}_A$ and contracting the auxiliary space $\mathcal{H}_{B_{aux}}$ of Bob with a vector'. Whether this sheds light on a general construction for $\tau$ is an open question. 

\subsection{The hierarchy as a tool to bound quantum violations of steering inequalities}

\blk

One of the main technical uses of the hierarchy of semidefinite tests that we have just defined is to bound the quantum violation of steering inequalities. The first use of our hierarchy for these purposes was presented in Ref.~\cite{bpqs}, and here we recall their results. 

Consider a Bob-with-Input scenario with $|\A|=2$, $|\X|=3$, and $|\Y|=2$. Now consider the steering inequality 
\begin{align}
I[\As] = \Tr{\sum_{a\in\A,x\in\X,y\in\Y} \mathcal{F}_{axy} \, \sigma_{a|xy} }
\end{align}
with $\mathcal{F}_{axy} = \frac{1}{2} (\id - (-1)^a \sigma_x)^{T^y} $, where $(\sigma_1,\sigma_2,\sigma_3) = (X,Y.Z)$ are the Pauli operators, and $T^y$ denotes the transpose operation when $y=1$ (and do nothing otherwise). 

The value of inequality is bounded from below, as shown in Ref.~\cite{bpqs}. Indeed, the minimum value it can attain for non-signalling assemblages is $0$, whilst the minimum value it can attain for LHS assemblages is $1.2679$. The minimum value attainable with quantum assemblages is not known, but a lower bound for it can be set by using our hierarchy of semidefinite tests. As shown in Ref.~\cite{bpqs}, the first level of our hierarchy yields a value:
\begin{align}
\min_{\As \in \mathcal{Q}^1} \, I[\As] = 0.4135 \,.
\end{align}
Therefore, whenever $\As$ is a quantum assemblage it will evaluate the steering functional to $I[\As] \geq 0.4135 $. Conversely, whenever $I[\As] < 0.4135 $ it follows that $\As$ does not admit of a quantum realisation, i.e., it is post-quantum. 

\blk

\section{Multipartite Bob-with-Input EPR scenarios}

We now generalise our hierarchy of semidefinite tests and results to multipartite \blk Bob-with-Input \blk EPR scenario. The outline of this section follows similarly to the one for bipartite scenarios.

\subsection{Definition of this EPR scenario}\label{se:bwiMul}

A natural question to ask is what happens if the Bob-with-Input EPR scenario includes more than one black-box party, i.e., many `Alices' (see Fig.~\ref{fig:mbwi}). This is a straightforward generalisation of the original Bob-with-Input EPR scenario, by incorporating the formal elements of multipartite traditional EPR scenarios \cite{oform,chan}. In this section we will specify these multipartite Bob-with-Input scenario by presenting, for simplicity, the case of two black-box parties  (Alice${}_1$ and Alice${}_2$). The case for more than two Alices follows straightforwardly. 

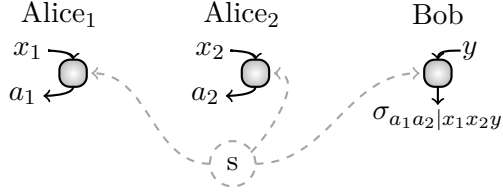
\begin{figure}
\begin{center}
	    \begin{tikzpicture}[scale=1.2]
		\node at (-6.3,1.3) {Alice${}_1$};
		\shade[draw, thick, ,rounded corners, inner color=white,outer color=gray!50!white] (-5.7,-0.3) rectangle (-6.3,0.3) ;
		\draw[thick, ->] (-6.5,0.5) to [out=180, in=90] (-6,0.3);
		\node at (-6.99,0.5) {$x_1$};
		\draw[thick, -<] (-6,-0.3) to [out=-90, in=180] (-6.5,-0.5);
		\node at (-7.1,-0.5) {$a_1$};

		\node at (-2.3,1.3) {Alice${}_2$};
		\shade[draw, thick, ,rounded corners, inner color=white,outer color=gray!50!white] (-1.7,-0.3) rectangle (-2.3,0.3) ;
		\draw[thick, ->] (-2.5,0.5) to [out=180, in=90] (-2,0.3);
		\node at (-2.99,0.5) {$x_2$};
		\draw[thick, -<] (-2,-0.3) to [out=-90, in=180] (-2.5,-0.5);
		\node at (-3.1,-0.5) {$a_2$};

		\node at (2,1.3) {Bob};
		\shade[draw, thick, ,rounded corners, inner color=white,outer color=gray!50!white] (1.7,-0.3) rectangle (2.3,0.3) ;		
		\draw[thick, ->] (2.5,0.5) to [out=180, in=90] (2,0.3);
		\node at (2.7,0.5) {$y$};
		\draw[thick, ->] (2,-0.3) -- (2,-0.7);
		\node at (2,-1) {$\sigma_{a_1a_2|x_1x_2y}$};
		
		\node at (-2.5,-2) {s};
		\draw[thick, dashed, color=gray!70!white] (-2.5,-2) circle (0.5cm);
		\draw[thick, dashed, color=gray!70!white, ->] (-3,-2) to [out=180, in=0] (-5.6,0);
		\draw[thick, dashed, color=gray!70!white, ->] (-2.1,-1.75) to [out=45, in=0] (-1.6,0);
		\draw[thick, dashed, color=gray!70!white, ->] (-2,-2) to [out=0, in=180] (1.6,0);		
	    \end{tikzpicture}	    
\end{center}
\caption{Multipartite Bob-with-Input steering scenario: example where two black-box parties (Alice${}_1$ and Alice${}_2$) steer the state of Bob's system, and where Bob has an input which is used to determine the production of a his quantum system.}
\label{fig:mbwi}
\end{figure}

Similarly to Section \ref{se:bwi}, for $k \in \{1,2\}$ let $x_k\in \{1, \ldots, m_A\} =: \X$ denote the choice of Alice${}_k$'s measurement, $a_k \in \{0, \ldots, o_A-1\} =: \A$ denote the outcome of Alice${}_k$'s measurement\footnote{For simplicity in the presentation we will take the sets $\X$ and $\A$ to be the same for all the Alices. The case when these sets are different follows directly.}, $y \in \{1, \ldots, m_B\} =: \Y$ be the choice of Bob's input, and $d$ the dimension of the Hilbert space of the quantum system whose state is prepared by this protocol. Then, the assemblage $\As_2$ (where $2$ refers to the number of Alices) consists of the following elements: $\{\sigma_{a_1a_2|x_1x_2y}\}_{a_1,a_2 \in \A, \, x_1,x_2 \in \X, \, y \in \Y} \equiv \As_2$. We see then that a multipartite Bob-with-Input EPR scenario is specified by the number of Alices, the dimension $d$, and the cardinalities of the sets $\A$, $\X$, and $\Y$. 

Similarly to the bipartite case, in these multipartite setups the parties -- which are distant -- are assumed to perform space-like separated actions on their share of the system. At the level of the assemblages $\As_2$, this is captured by the so-called No Signalling condition that the elements of $\As_2$ must satisfy. In a nutshell, (i) no Alice can signal any other party by the means of her measurement choice (i.e., $\sum_{a_1 \in \A} \sigma_{a_1a_2|x_1x_2y} = \sum_{a_1 \in \A} \sigma_{a_1a_2|x_1'x_2y}$ for all $x_1,x'_1,x_2 \in \X$ and $y \in \Y$, and similarly when summing over Alice$_{2}$'s outcomes instead), and (ii) Bob cannot signal to any of the Alices by means of his choice of $y$ (i.e., $\tr{\sigma_{a_1a_2|x_1'x_2y}}{}=\tr{\sigma_{a_1a_2|x_1'x_2y'}}{}\,, \, \forall \, a_1,a_2 \in \A \,,\, x_1,x_2\in\X\,,\, y,y'\in\Y$). The assemblages $\As_2$ that such an EPR scenario admits are therefore known as \textit{non-signalling assemblages}.

An assemblage $\As_2$ is then defined to have a quantum realisation in the Multipartite Bob-with-Input EPR scenario if the following conditions are satisfied: 
\begin{defn}\textbf{Quantum multipartite Bob-with-Input assemblages.--}\label{def:qmbwia}\\
An assemblage $\As_2 = \{\sigma_{a_1a_2|x_1x_2y}\}_{a_1,a_2 \in \A, \, x_1,x_2 \in \X, \, y \in \Y}$ in a Multipartite Bob-with-Input EPR scenario with two Alices has a quantum realisation if there exist:
\begin{compactitem}
\item For each Alice: a Hilbert space $\cH_k$ ($k \in \{1,2\}$),
\item For each Alice: a collection of $m_A$ projective measurements of $o_A$ outcomes each on $\cH_k$ -- $\{\Pi^{(k)}_{a_k|x_k}\}_{a_k \in \A}$ for each $x_k \in \X$, 
\item A tripartite quantum system $\rho$ in the Hilbert space $\cH_1 \otimes \cH_2 \otimes \cH_{\mathrm{B}_\mathrm{(i)}}$,
\item a CPTP map $\cE_y$ ($ \cH_{\mathrm{B}_\mathrm{(i)}} \rightarrow  \cH_{\mathrm{B}}$) for Bob for each $y \in \Y$,
\end{compactitem}
such that
\begin{align}
\sigma_{a_1a_2|x_1x_2y} = \cE_y\left[\tr{\Pi^{(1)}_{a_1|x_1} \otimes \Pi^{(2)}_{a_2|x_2} \otimes \id_B  \, \rho}{A_1A_2}\right] \quad \forall \, a_1,a_2 \in \A, \, x_1,x_2 \in \X, \, y \in \Y\,.
\end{align}
The set of all quantum realisable assemblages in a specified multipartite Bob-with-Input EPR scenario is denoted by $\cQ_{2}$, where here $2$ refers to the number of black-box parties\footnote{In truth, the notation for the set should also include the parameters $(d,|\A|,|\X|,|\Y|)$. Since this information will be clear from context in the current manuscript, we will omit such subscripts for clarity in the presentation.}.
\end{defn}

In the following section we show how to generalise to these multipartite scenarios the hierarchy defined in Section \ref{se:BwIhier}.

\subsection{A hierarchy for the multipartite Bob-with-Input scenario}

Here we  define a hierarchy of sets of assemblages in the multipartie Bob-with-Input EPR scenario, generalising the ideas of Section \ref{se:BwIhier}. Hereon, the number of black-box parties in the scenario will be denoted by $\Nb$. 

The alphabet $\Upsilon_{\Nb}$ that we  use to define the words is composed of the following elements: 
\begin{align}
\Upsilon_{\Nb} := \{\emptyset\} \cup \{ a_1|x_1 \}_{{ x_1 \in \X\,, a_1 \in \A \setminus \{0\}}} \cup \ldots \cup \{ a_\Nb |x_\Nb \}_{{ x_\Nb \in \X\,, a_\Nb \in \A \setminus \{0\}}} \cup \{ y \}_{y \in \Y} \,. 
\end{align}

The definition of words, their concatenation, and their ``dagger'' follows similarly from Sec.~\ref{se:BwIhier}. We denote by $\cS^*_\Nb$ the set of all words of arbitrary length with letters drawn from $\Upsilon_\Nb$.

The symmetry operations on the elements of the alphabet generalise to the following:
\begin{compactitem} 
\item $\mathbf{v}\mathbf{w} \equiv \mathbf{v} \emptyset \mathbf{w} $ for all $\mathbf{u},\mathbf{w} \in \cS^*_\Nb$,
\item $\mathbf{v}\mathbf{v} \equiv \mathbf{v}$ for all $\mathbf{v} \in \Upsilon_\Nb$,
\item $a_k|x_k y \equiv y a_k|x_k$ for all $y,a_k|x_k \in \Upsilon_\Nb$ and $k \in \{1,\ldots,\Nb\}$,
\item $a_k|x_k a_{k'}|x_{k'}  \equiv a_{k'}|x_{k'} a_k|x_k$ for all $a_k|x_k,a_{k'}|x_{k'} \in \Upsilon_\Nb$ and $k\neq k' \in \{1,\ldots,\Nb\}$.
\end{compactitem} 
Two words are then said to be equivalent if there exists a sequence of symmetry operations that can take one to the other, and vice-versa, just like in Sec.~\ref{se:BwIhier}. Moreover, a word $\mathbf{v}$ is now \textit{null} if, after applying a sequence of symmetry operations, one may find a letter $a_k|x_k$ followed by a letter $a'_k|x_k$ for some $x_k \in \X$, $a_k \neq a'_k \in \A \setminus \{0\}$, and $k \in \{1,\ldots,\Nb\}$.

With all this, we can now generalise to multipartite scenarios the definitions of ``set of words of certain length'' and ``moment matrix of order $n$'':
\begin{defn} \textbf{The set of words} $\boldsymbol{\cS_n^\Nb}$.--\\
A word $\mathbf{v} \in \cS^*_\Nb$ belongs to $\cS_n^\Nb$ if it may arise from concatenating at most $n$ letters of the alphabet $\Upsilon_\Nb$.
\end{defn}

\begin{defn}\label{def:MMnM}\textbf{Moment matrix of order $\mathbf{n}$ with $\Nb$ Alices: $\boldsymbol{\Gamma_n^{(\Nb)}}$}.--\\
Let $\Gamma_n^{(\Nb)}$ be a matrix of size $\lvert \cS_n^\Nb \lvert \times \lvert \cS_n^\Nb \lvert$, whose entries are $d \times d$ complex matrices, and whose rows and columns are indexed by the words in $\cS_n^\Nb$. This matrix is a moment matrix of order $n$ iff it satisfies the following properties: 
\begin{align}
\Gamma_n^{(\Nb)} \geq 0 \,, \label{eq:2ag}\\
\Gamma_n^{(\Nb)}(\emptyset,\emptyset) &= \id_d \,,\label{eq:3ag}\\
\Gamma_n^{(\Nb)}(\mathbf{v},\mathbf{w}) &= \Gamma^\Nb_n(\mathbf{v'},\mathbf{w'}) \quad \text{if} \quad \mathbf{v}^\dagger\mathbf{w} \equiv \mathbf{v'}^\dagger\mathbf{w'} \,,\label{eq:4ag}\\
\Gamma_n^{(\Nb)}(\mathbf{v},\mathbf{w}) &= \mathbf{0}_d  \quad \text{if} \quad \mathbf{v}^\dagger\mathbf{w}  \, \text{is null} \,,\label{eq:5ag}\\
\Gamma_n^{(\Nb)}(a_k|x_k,a'_{k'}|x'_{k'}) &\propto \id_d \quad \text{for all} \quad a_k|x_k,a'_{k'}|x'_{k'} \in \Upsilon_\Nb \quad \text{with} \quad  k,k' \in \{1,\ldots,\Nb\} \,,  \label{eq:6ag}
\end{align}
where $\mathbf{0}_d$ is the $d \times d$ matrix whose all entries are 0. 
\end{defn}

Finally, the definitions of ``Certificate of order $n$ for an assemblage'' and ``Set of assemblages in level $n$ of the hierarchy'' generalise as follows: 

\begin{defn}\label{def:certNb} \textbf{Certificate of order $n$ for an assemblage $\As_\Nb$}.--\\
Let $\As_\Nb = \{\sigma_{a_1 \ldots a_\Nb|x_1 \ldots x_\Nb y}\}_{a_k \in \A, \, x_k \in \X, \, k \in \{1,\ldots,\Nb\}, \, y \in \Y}$ be a non-signalling assemblage in the multipartite Bob-with-Input EPR scenario with $\Nb$ black-box parties.

Let $\Omega \subset \{1,\ldots,\Nb\}$ be a subset of the black-box parties, and denote by $\sigma_{\mathbf{a}_\Omega | \mathbf{x}_\Omega y}$ the assemblage elements that result when tracing out the parties not in $\Omega$, i.e.,
\begin{align}\label{eq:32}
\sigma_{\mathbf{a}_\Omega | \mathbf{x}_\Omega y} \stackrel{\text{def}}{=} \sum_{j \not\in \Omega}   \sum_{\phantom{a} a_j\in \A} \sigma_{a_1 \ldots a_\Nb|x_1 \ldots x_\Nb y}\,.
\end{align}

A matrix $\Gamma_n^{(\Nb)}$ is a certificate of order $n$ for $\As_\Nb$ iff $\Gamma_n^{(\Nb)}$ is a moment matrix of order $n$ for $\Nb$ Alices, and the following are satisfied: 
\begin{align}
\Gamma_n^{(\Nb)}(\emptyset,\mathbf{a}_\Omega | \mathbf{x}_\Omega) &= \Tr{\sigma_{\mathbf{a}_\Omega | \mathbf{x}_\Omega 1} } \, \id_d \quad \forall \, a_k|x_k \in \Upsilon_\Nb\,, \quad k \in \Omega\,,\quad \forall \Omega\,, \label{eq:7ag}\\
\Gamma_n^{(\Nb)}(\emptyset,y) &= \frac{1}{d} \, \sigma^\mathrm{T}_{y} \quad \forall \, y \in \Upsilon_\Nb \,,\label{eq:8ag}\\
\Gamma_n^{(\Nb)}(\emptyset, \mathbf{a}_\Omega | \mathbf{x}_\Omega \, y) &= \frac{1}{d} \, \sigma^\mathrm{T}_{\mathbf{a}_\Omega | \mathbf{x}_\Omega y} \quad \forall a_k|x_k,y \in \Upsilon_\Nb\,, \quad k \in \Omega\,,\quad \forall \Omega\,,\label{eq:9ag}
\end{align}
where $\sigma_y = \sum_{a_1\in\A\,,\ldots \,, a_\Nb\in\A}  \sigma_{a_1 \ldots a_\Nb|x_1 \ldots x_\Nb y}$ is the marginal state of Bob's system.
\end{defn}

\begin{defn}\textbf{Set of assemblages in level $n$ of the multipartite hierarchy: $\cQ^n_\Nb$}.--\\
An assemblage $\As_\Nb$ in the multipartite Bob-with-Input scenario with $\Nb$ black-box parties belongs to the set $\cQ^n_\Nb$ iff there exists a matrix $\Gamma_n^{(\Nb)}$ that is a certificate of order $n$ for $\As_\Nb$. 
\end{defn}

Next we will show that quantum assemblages in the multipartite Bob-with-Input scenario belong to $\cQ^n_\Nb$ for all $n$ and $\Nb$. 

\subsection{Multipartite Bob-with-Input quantum assemblages are $\cQ^n_\Nb$ assemblages}\label{multibwi}

\begin{thm}\textbf{$\cQ_\Nb \subset \cQ^n_\Nb \,, \quad \forall \, n \in \mathbb{N}\,, \forall \, \Nb$.}
\end{thm}
\begin{proof}
The proof of the claim follows a straight parallelism to that of Sec.~\ref{se:3}. That is, we will take an arbitrary quantum assemblage $\As^{\mathrm{Q}}_\Nb \in \cQ_\Nb$ and show that $\As^{\mathrm{Q}}_\Nb \in \cQ^n_\Nb \,, \, \forall \, n \in \mathbb{N}\,, \forall \, \Nb$. 

Let $\ket{\psi}\bra{\psi} = \rho'$, $\{\Pi_{a_k|x_k}\}_{a_k \in \A\,,\, x_k \in \X\,,\, k \in \{1, \ldots, \Nb\}}$, and $\{U_y\}_{y \in \Y}$ be a pure state, projective measurements, and unitary transformations (dilations of the CPTP maps $\cE_y$) which realise the quantum assemblage $\As^{\mathrm{Q}}_\Nb$, similarly to Eq.~\eqref{eq:therhop}:

\begin{align}\label{eq:therhopm}
\sigma_{a_1 \ldots a_\Nb|x_1 \ldots x_\Nb y} \quad = \quad %
\begin{tikzpicture}
	\begin{pgfonlayer}{nodelayer}
		\node [style=none] (0) at (-2, -1) {$\rho^\prime$};
		\node [style=none] (1) at (-7, -0.5) {};
		\node [style=none] (2) at (-2, -2) {};
		\node [style=none] (3) at (3, -0.5) {};
		\node [style=none] (4) at (-5.7, 1.5) {\tiny{$\Pi_{a_1|x_1}$}};
		\node [style=none] (5) at (-5.75, -0.5) {};
		\node [style=none] (7) at (2.25, -0.5) {};
		\node [style=none] (8) at (2.25, 3.25) {};
		\node [style=right label] (9) at (-5.75, 0.25) {$A_1$};
		\node [style=right label] (10) at (2.25, 0.25) {$B_{(i)}$};
		\node [style=right label] (11) at (2.25, 2.75) {$B$};
		\node [style=none] (12) at (-7.25, 1) {};
		\node [style=none] (13) at (-4.25, 1) {};
		\node [style=none] (14) at (-5.75, 3) {};
		\node [style=none] (15) at (-5.75, 1) {};
		\node [style=none] (16) at (0.5, 1) {};
		\node [style=none] (17) at (3, 1) {};
		\node [style=none] (18) at (3, 2) {};
		\node [style=none] (19) at (0.5, 2) {};
		\node [style=none] (20) at (1.75, 1.5) {$U_y$};
		\node [style=none] (21) at (2.25, 2) {};
		\node [style=none] (22) at (2.25, 1) {};
		\node [style=none] (23) at (1.25, -0.5) {};
		\node [style=left label] (24) at (1.25, 0.25) {$B_{aux}$};
		\node [style=none] (25) at (1.25, 1) {};
		\node [style=none] (26) at (1.25, 2.5) {};
		\node [style=none] (28) at (1.25, 2) {};
		\node [style=upground] (29) at (1.25, 2.75) {};
		\node [style=none] (30) at (-1.65, 1.5) {\tiny{$\Pi_{a_\Nb|x_\Nb}$}};
		\node [style=none] (31) at (-1.75, -0.5) {};
		\node [style=right label] (32) at (-1.75, 0.25) {$A_\Nb$};
		\node [style=none] (33) at (-3.25, 1) {};
		\node [style=none] (34) at (-0.25, 1) {};
		\node [style=none] (35) at (-1.75, 3) {};
		\node [style=none] (36) at (-1.75, 1) {};
		\node [style=none] (37) at (-3.75, 1) {$\ldots$};
	\end{pgfonlayer}
	\begin{pgfonlayer}{edgelayer}
		\draw (1.center) to (2.center);
		\draw (2.center) to (3.center);
		\draw (3.center) to (1.center);
		\draw (14.center) to (12.center);
		\draw (12.center) to (13.center);
		\draw (13.center) to (14.center);
		\draw [qWire] (15.center) to (5.center);
		\draw (19.center) to (16.center);
		\draw (16.center) to (17.center);
		\draw (17.center) to (18.center);
		\draw (18.center) to (19.center);
		\draw [qWire] (21.center) to (8.center);
		\draw [qWire] (22.center) to (7.center);
		\draw [qWire] (25.center) to (23.center);
		\draw [qWire] (28.center) to (26.center);
		\draw (35.center) to (33.center);
		\draw (33.center) to (34.center);
		\draw (34.center) to (35.center);
		\draw [qWire] (36.center) to (31.center);
	\end{pgfonlayer}
\end{tikzpicture}
}\,.
\end{align}

The first step is to construct the projective measurement $\{\Psi_{1|y}, \Psi_{2|y} \}$ for Bob for each $y \in \Y$, just as in Eq.~\eqref{eq:bobmeas}.

Next, define the set of operators $\{\Op_\mathbf{v}\}_{\mathbf{v} \in \cS_n^\Nb}$, labelled by the words of the set $\cS_n^\Nb$, which act on the Hilbert space $\cH= \left(\bigotimes_{k \in \{1,\ldots,\Nb\}} \cH_{A_k} \right) \otimes \cH_{B_{aux}} \otimes \cH_{B_{(i)}}\otimes \cH_{B'}$. For the case of single-letter words (i.e., $\mathbf{v} \in \Upsilon_\Nb$), the operators are defined as:
\begin{align}
\Op_\emptyset &:= \id_{A_1 \ldots A_\Nb B_{aux}B_{(i)}B'} \,,\\
\Op_{a_k|x_k} &:= \left(\bigotimes_{j < k} \id_{A_j} \right)  \otimes \Pi_{a_k|x_k} \otimes \left(\bigotimes_{j > k} \id_{A_j} \right) \, \otimes \id_{B_{aux}B_{(i)}B'} \quad \forall \, a_k|x_k \in \Upsilon_\Nb \,, \, k \in \{1,\ldots,\Nb\} \,,\\
\Op_{y} &:=  \left(\bigotimes_{k=1:\Nb} \id_{A_k} \right) \, \otimes \Psi_{1|y} \quad \forall \, y \in \Upsilon_\Nb\,.
\end{align}
For the case of words with multiple letters, i.e., $\mathbf{v} = \mathbf{v}_1 \ldots \mathbf{v}_r$ with $1 < r \leq n$ and $\mathbf{v}_j \in \Upsilon_\Nb \quad \forall 1 \leq j \leq r$, the operators are defined as:
\begin{align}
\Op_{\mathbf{v}} := \Op_{\mathbf{v}_1} \ldots \Op_{\mathbf{v}_r} \,.
\end{align}
Hence, 
\begin{align}
\Op_{\mathbf{v}} \Op_{\mathbf{w}} = \Op_{\mathbf{vw}}.
\end{align}
Notice also that:
\begin{align}
\Op_\emptyset \Op_{\mathbf{v}} &= \Op_{\mathbf{v}} = \Op_{\mathbf{v}} \Op_{\emptyset} &\forall& \, \mathbf{v} \in \cS_{n-1}^\Nb \,,\\
\Op_{a_k|x_k} \Op_y &= \Op_y \Op_{a_k|x_k}  &\forall& \, a_k|x_k , y \in \Upsilon_\Nb\,, \, k \in \{1,\ldots,\Nb\} \,, \,\\
\Op_{a_k|x_k} \Op_{a_{k'}|x_{k'}} &= \Op_{a_{k'}|x_{k'}} \Op_{a_k|x_k} \, &\forall& \, a_k|x_k , a_{k'}|x_{k'} \in \Upsilon_\Nb\,, \, k \neq k' \in \{1,\ldots,\Nb\} \,, \,\\
\Op_{\mathbf{v}} \Op_{\mathbf{v}} &= \Op_{\mathbf{v}} \quad &\forall& \, \mathbf{v} \in \Upsilon_\Nb\,,
\end{align} 
therefore 
\begin{align}
\mathbf{v} \equiv \mathbf{w} \, \Rightarrow \, \Op_{\mathbf{v}} = \Op_{\mathbf{w}}.
\end{align}
In addition, for any $k \in \{1, \ldots , \Nb \}$, $\Op_{a_k|x_k} \Op_{a'_k|x_k} = \mathbf{0}_{\cH}$ (i.e., the zero matrix in the Hilbert space $\cH$) whenever $a_k\neq a'_k$ for all $x_k \in \X$. Therefore
\begin{align}
\Op_{\mathbf{v}} =  \mathbf{0}_{\cH} \quad \text{when} \,\, \mathbf{v} \,\, \text{is null} \,.
\end{align}
Finally, since the operators for single-lettered words are Hermitian, it follows that:
\begin{align}
\Op_\mathbf{v}^\dagger = \Op_{\mathbf{v}^\dagger}\,.
\end{align}

Now, define the quantum state $\tilde{\rho} = \ket{\psi}\bra{\psi} \otimes \id_{B'}$, and with it define a square matrix $\Gamma^\Nb_n$ of dimension $|\cS^\Nb_n|$, whose rows and columns are indexed by the words in $\cS_n^\Nb$, as:
\begin{align}
\Gamma^\Nb_n(\mathbf{v},\mathbf{w}) := \tr{\Op_{\mathbf{v}}^\dagger \, \Op_{\mathbf{w}} \, \tilde{\rho}}{A_1 \ldots A_\Nb B_{aux}B_{(i)}}\,.
\end{align}
 Next we will show that $\Gamma_n^\Nb$ is a certificate of order $n$ for the assemblage $\As^{\mathrm{Q}}_\Nb$.

Let us first begin by showing that $\Gamma_n^\Nb$ is indeed a moment matrix of order $n$ for $\Nb$ Alices. For this, we need to show that its element satisfy the five sets of conditions listed in Definition \ref{def:MMnM} via Eqs.~\eqref{eq:2ag} to \eqref{eq:6ag}:
\begin{compactitem}
\item The $(i,j)$ entry of the element $\Gamma^\Nb_n(\mathbf{v},\mathbf{w})$ is 
\begin{align*}
\Gamma_n^{\Nb (i,j)}(\mathbf{v},\mathbf{w}) &= \bra{i} \tr{\Op_{\mathbf{v}}^\dagger \, \Op_{\mathbf{w}} \, \tilde{\rho}}{A_1 \ldots A_\Nb B_{aux}B_{(i)}} \ket{j} \\
&=\bra{i} \tr{\Op_{\mathbf{v}}^\dagger \, \Op_{\mathbf{w}} \, \ket{\psi} \bra{\psi} \otimes \id_d^{B'} }{A_1 \ldots A_\Nb B_{aux}B_{(i)}} \ket{j} \\
&= \bra{i} \bra{\psi} \Op_{\mathbf{v}}^\dagger \, \Op_{\mathbf{w}} \ket{\psi} \ket{j}\,.
\end{align*}
Now define the vectors $\ket{\psi_{i,\mathbf{v}}} = \Op_{\mathbf{v}} \ket{\psi} \ket{i}$. Hence, $\Gamma_n^{\Nb (i,j)}(\mathbf{v},\mathbf{w}) = \langle \psi_{i,\mathbf{v}} \lvert \psi_{i,\mathbf{w}} \rangle$. This shows that $\Gamma^\Nb_n$ is a Gramian matrix, and therefore positive semidefinite. Hence, Eq.~\eqref{eq:2ag} is satisfied.
\item $\Gamma^\Nb_n(\emptyset,\emptyset) = \tr{\id_\cH \, \id_\cH \, \tilde{\rho}}{A_1 \ldots A_\Nb B_{aux}B_{(i)}} =  \id_{B'}$, hence Eq.~\eqref{eq:3ag} is satisfied.
\item On the one hand, $\Gamma^\Nb_n(\mathbf{v},\mathbf{w}) = \tr{\Op_{\mathbf{v}}^\dagger \, \Op_{\mathbf{w}} \, \tilde{\rho}}{A_1 \ldots A_\Nb B_{aux}B_{(i)}} = \tr{\Op_{\mathbf{v}^\dagger\mathbf{w}} \, \tilde{\rho}}{A_1 \ldots A_\Nb B_{aux}B_{(i)}}$. \\
On the other hand, $\Gamma^\Nb_n(\mathbf{v}',\mathbf{w}') = \tr{\Op_{\mathbf{v}'}^\dagger \, \Op_{\mathbf{w}'} \, \tilde{\rho}}{A_1 \ldots A_\Nb B_{aux}B_{(i)}} = \tr{\Op_{\mathbf{v}'^\dagger\mathbf{w}'} \, \tilde{\rho}}{A_1 \ldots A_\Nb B_{aux}B_{(i)}}$. \\ If $\mathbf{v}^\dagger\mathbf{w} \equiv \mathbf{v}'^\dagger\mathbf{w}'$, then $\Op_{\mathbf{v}^\dagger\mathbf{w}} = \Op_{\mathbf{v}'^\dagger\mathbf{w}'}$, and hence $\Gamma^\Nb_n(\mathbf{v},\mathbf{w}) = \Gamma^\Nb_n(\mathbf{v}',\mathbf{w}')$. It follows that Eq.~\eqref{eq:4ag} is satisfied.
\item If $\mathbf{v}^\dagger\mathbf{w}$ is null, then $\Op_{\mathbf{v}^\dagger\mathbf{w}} =  \mathbf{0}_{\cH} $. Hence, $\Gamma^\Nb_n(\mathbf{v},\mathbf{w}) = \tr{\Op_{\mathbf{v}^\dagger\mathbf{w}} \, \tilde{\rho}}{A_1 \ldots A_\Nb B_{aux}B_{(i)}} = \mathbf{0}_d$, and so Eq.~\eqref{eq:5ag} is satisfied. 
\item For any $a_k|x_k, a'_{k'}|x'_{k'} \in \Upsilon_\Nb$, define
\begin{align*}
\tilde{\Pi}_{a_k|x_k} := \left(\bigotimes_{j < k} \id_{A_j} \right)  \otimes \Pi_{a_k|x_k} \otimes \left(\bigotimes_{j > k} \id_{A_j} \right)\,,\\
\tilde{\Pi}_{a'_{k'}|x'_{k'}} := \left(\bigotimes_{j < k'} \id_{A_j} \right)  \otimes \Pi_{a'_{k'}|x'_{k'}} \otimes \left(\bigotimes_{j > k'} \id_{A_j} \right) \,.
\end{align*}
Hence,
\begin{align*}
\Gamma^\Nb_{a_k|x_k, a'_{k'}|x'_{k'}} &= \tr{ \left(  \tilde{\Pi}_{a_k|x_k} \,  \tilde{\Pi}_{a'_{k'}|x'_{k'}} \right) \otimes \id_{B_{aux}B_{(i)}B'}   \, \tilde{\rho}}{A_1 \ldots A_\Nb B_{aux}B_{(i)}} \\
&= \tr{ \left(   \tilde{\Pi}_{a_k|x_k} \,  \tilde{\Pi}_{a'_{k'}|x'_{k'}}  \right) \otimes \id_{B_{aux}B_{(i)}}   \, \ket{\psi}\bra{\psi} }{} \, \id_{B'} \,.
\end{align*}
Hence, $\Gamma^\Nb_{a|x, a'|x'} \propto \id_d$ and Eq.~\eqref{eq:6ag} follows. 
\end{compactitem}
Now that we've shown that $\Gamma^\Nb_n$ is indeed a moment matrix of order $n$, we need to show that it is moreover a certificate of order $n$ for the assemblage $\As_\Nb^{\mathrm{Q}}$. For this, we need to now show that $\Gamma^\Nb_n$ satisfies the conditions given by Eqs.~\eqref{eq:7ag} to \eqref{eq:9ag} in Definition \ref{def:certNb}. In the following, let $\Omega \subset \{1,\ldots, \Nb\}$ be a subset of the black-box parties.
\begin{compactitem}
\item Consider an arbitrary word $\mathbf{a}_\Omega | \mathbf{x}_\Omega$ composed by concatenating letters $a_{k}|x_{k} \in \Upsilon_\Nb$ where $k \in \Omega$. Define
\begin{align*}
\tilde{\Pi}_{\mathbf{a}_\Omega | \mathbf{x}_\Omega} := \bigotimes_{k=1:\Nb} \left( \id_{A_k} \delta_{k \not\in \Omega} + \Pi_{a_k|x_k} \delta_{k \in \Omega}  \right)\,.
\end{align*}
That is, $\tilde{\Pi}_{\mathbf{a}_\Omega | \mathbf{x}_\Omega}$ is a projector in the Hilbert space of the Alices, constructed as a product of operators for each Alice, where the local operator is $\id_{A_k}$ when $k \not\in \Omega$, and $\Pi_{a_k|x_k}$ when $k \in \Omega$ for the corresponding letter $a_k|x_k$ that appears in $\mathbf{a}_\Omega | \mathbf{x}_\Omega$. 
\item[] Hence,
\begin{align*}
\Gamma^\Nb_n(\emptyset,\mathbf{a}_\Omega | \mathbf{x}_\Omega) &= \tr{\id_\cH \, \tilde{\Pi}_{\mathbf{a}_\Omega | \mathbf{x}_\Omega} \otimes \id_{B_{aux}B_{(i)}B'} \, \tilde{\rho}}{A_1 \ldots A_\Nb B_{aux}B_{(i)}} \\
&= \tr{\tilde{\Pi}_{\mathbf{a}_\Omega | \mathbf{x}_\Omega} \otimes \id_{B_{aux}B_{(i)}B'} \, \ket{\psi}\bra{\psi} \otimes \id_{B'}}{A_1 \ldots A_\Nb B_{aux}B_{(i)}} \\
&= \tr{\tilde{\Pi}_{\mathbf{a}_\Omega | \mathbf{x}_\Omega} \otimes \id_{B_{aux}B_{(i)}} \, \ket{\psi}\bra{\psi}}{} \,  \id_{B'} \\
&= \tr{\tilde{\Pi}_{\mathbf{a}_\Omega | \mathbf{x}_\Omega} \otimes \id_{B_{aux}B_{(i)}} \, (\id_{A_1 \ldots A_\Nb} \otimes U_y) \ket{\psi}\bra{\psi} (\id_{A_1 \ldots A_\Nb} \otimes U_y^\dagger) }{} \,  \id_{B'}\,,
\end{align*}
since unitary operations are trace preserving. 
\item[] Notice also that $\sum_{a_r \in \A} \Pi_{a_k|x_k} = \id_{A_k}$ $\forall \, x_k\in\X\,, k\in\{1,\ldots,\Nb\}$. Hence,  
\begin{align*}
\tilde{\Pi}_{\mathbf{a}_\Omega | \mathbf{x}_\Omega} = \sum_{j \not\in \Omega} \, \sum_{a_j \in \A} \, \Pi_{a_1|x_1} \otimes \ldots \otimes \Pi_{a_\Nb |x_\Nb }\,,
\end{align*}
for any choice of $x_j \in \X$ for $j \not\in \Omega$. It follows that
\begin{align*}
&\Gamma^\Nb_n(\emptyset,\mathbf{a}_\Omega | \mathbf{x}_\Omega) = \tr{\tilde{\Pi}_{\mathbf{a}_\Omega | \mathbf{x}_\Omega} \otimes \id_{B_{aux}B_{(i)}} \, (\id_{A_1 \ldots A_\Nb} \otimes U_y) \ket{\psi}\bra{\psi} (\id_{A_1 \ldots A_\Nb} \otimes U_y^\dagger) }{} \,  \id_{B'}\,, \\
&= \sum_{j \not\in \Omega} \, \sum_{a_j \in \A} \, \tr{\Pi_{a_1|x_1} \otimes \ldots \otimes \Pi_{a_\Nb |x_\Nb } \otimes \id_{B_{aux}B_{(i)}} \, (\id_{A_1 \ldots A_\Nb} \otimes U_y) \ket{\psi}\bra{\psi} (\id_{A_1 \ldots A_\Nb} \otimes U_y^\dagger) }{} \,  \id_{B'}\,, \\
&= \sum_{j \not\in \Omega} \, \sum_{a_j \in \A} \, \tr{\sigma_{a_1\ldots a_\Nb | x_1 \ldots x_\Nb y}}{B}  \,  \id_{B'} \\
&= \tr{\sigma_{\mathbf{a}_\Omega | \mathbf{x}_\Omega y}}{}  \,  \id_{B'}\,.
\end{align*}
Since, by definition, $ \tr{\sigma_{\mathbf{a}_\Omega | \mathbf{x}_\Omega y}}{} = \tr{\sigma_{\mathbf{a}_\Omega | \mathbf{x}_\Omega 1}}{}$, Eq.~\eqref{eq:7ag} is satisfied. 

\item Consider an arbitrary word $\mathbf{a}_\Omega | \mathbf{x}_\Omega$ composed by concatenating letters $a_{k}|x_{k} \in \Upsilon_\Nb$ where $k \in \Omega$. For any such word $\mathbf{a}_\Omega | \mathbf{x}_\Omega$ and any  $y \in \Upsilon_\Nb$,
\begin{align*}
&\Gamma^\Nb_n(\emptyset,\mathbf{a}_\Omega | \mathbf{x}_\Omega y) = \tr{\id_\cH \, \tilde{\Pi}_{\mathbf{a}_\Omega | \mathbf{x}_\Omega} \otimes \Psi_{1|y} \, \tilde{\rho}}{A_1 \ldots A_\Nb B_{aux}B_{(i)}} \\
&= \tr{\tilde{\Pi}_{\mathbf{a}_\Omega | \mathbf{x}_\Omega}  \otimes \left( \left(U^\dagger_y \otimes \id_{B'}\right) \left( \id_{B_{aux}} \otimes  \Phi_1 \right) \left(U_y \otimes \id_{B'}\right) \right) \, \tilde{\rho}}{A_1 \ldots A_\Nb B_{aux}B_{(i)}} \\
&= \tr{\left( \tilde{\Pi}_{\mathbf{a}_\Omega | \mathbf{x}_\Omega}^\dagger \otimes  U^\dagger_y \otimes \id_{B'} \right)   \left( \id_{A_1 \ldots A_\Nb B_{aux}} \otimes  \Phi_1 \right) \left( \tilde{\Pi}_{\mathbf{a}_\Omega | \mathbf{x}_\Omega} \otimes  U_y \otimes \id_{B'} \right)    \, \tilde{\rho}    }{A_1 \ldots A_\Nb B_{aux}B_{(i)}} \\
&= \tr{  \left( \id_{A_1 \ldots A_\Nb B_{aux}} \otimes  \Phi_1 \right) \left( \tilde{\Pi}_{\mathbf{a}_\Omega | \mathbf{x}_\Omega} \otimes  U_y \otimes \id_{B'} \right)    \, \tilde{\rho}   \, \left( \tilde{\Pi}_{\mathbf{a}_\Omega | \mathbf{x}_\Omega}^\dagger \otimes  U^\dagger_y \otimes \id_{B'} \right)  }{A_1 \ldots A_\Nb B_{aux}B_{(i)}} \\
&= \tr{  \left( \id_{A_1 \ldots A_\Nb B_{aux}} \otimes  \Phi_1 \right) \left(  \left( \tilde{\Pi}_{\mathbf{a}_\Omega | \mathbf{x}_\Omega} \otimes  U_y \right)    \ket{\psi}\bra{\psi} \left( \tilde{\Pi}_{\mathbf{a}_\Omega | \mathbf{x}_\Omega}^\dagger \otimes  U^\dagger_y  \right) \otimes \id_{B'}  \right) }{A_1 \ldots A_\Nb B_{aux}B_{(i)}} \\
&= \sum_{j \not\in \Omega} \, \sum_{a_j \in \A} \,   \mathrm{tr}_{A_1 \ldots A_\Nb B_{aux}B_{(i)}} \Big\{  \left( \id_{A_1 \ldots A_\Nb B_{aux}} \otimes  \Phi_1 \right) \\ 
& \qquad  \left(  \left( \Pi_{a_1|x_1} \otimes \ldots \otimes \Pi_{a_\Nb |x_\Nb } \otimes  U_y \right)    \ket{\psi}\bra{\psi} \left( (\Pi_{a_1|x_1} \otimes \ldots \otimes \Pi_{a_\Nb |x_\Nb })^\dagger \otimes  U^\dagger_y  \right) \otimes \id_{B'}  \right) \Big\}
\end{align*}
for any choice of $x_j \in \X$ for $j \not\in \Omega$. Hence,
\begin{align*}
\Gamma^\Nb_n(\emptyset,\mathbf{a}_\Omega | \mathbf{x}_\Omega y) &= \sum_{j \not\in \Omega} \, \sum_{a_j \in \A} \,   \tr{ \Phi_1 \, \left( \sigma_{a_1\ldots a_\Nb | x_1 \ldots x_\Nb y} \otimes \id_{B'} \right) }{B} \\
&=  \tr{ \Phi_1 \, \left( \sigma_{  \mathbf{a}_\Omega | \mathbf{x}_\Omega y} \otimes \id_{B'} \right) }{B} \\
&= \frac{1}{d} \, \sigma_{\mathbf{a}_\Omega | \mathbf{x}_\Omega y}^{\mathrm{T}}\,.
\end{align*}
Hence, Eq.~\eqref{eq:9ag} follows. 
\item This last condition follows formally from the previous one, by taking $\Omega \equiv \emptyset$. Hence, the matrix $\Gamma^\Nb_n$ satisfies the conditions of Eq.~\eqref{eq:8ag}. 
\end{compactitem}

We have hence shown that the matrix $\Gamma^\Nb_n$ that we constructed is a certificate of order $n$ for the assemblage $\As_\Nb^{\mathrm{Q}}$, which implies that $\As_\Nb^{\mathrm{Q}} \in \cQ^n_\Nb$, for all $n \in \mathbb{N}$.
\end{proof}

\subsection{Convergence of the multipartite Bob-with-Input hierarchy}\label{se:convmul}

Let $\cQ_\Nb$ denote the set of quantumly-realisable assemblages in the multipatite Bob-with-Input scenario. In the previous section we showed that $\cQ_\Nb \subseteq \cQ^n_\Nb \,\, \forall \, n \in \mathbb{N}$. A corollary from Theorem \ref{thm:limitoperator} is then that the hierarchy $\{\cQ^n_\Nb\}_{ n \in \mathbb{N}}$ converges to a set of assemblages $\cQ_C^\Nb$ that therefore contains the quantum set.

\begin{cor} 
The hierarchy $\{\cQ^n_\Nb\}_{ n \in \mathbb{N}}$ converges to a set of assemblages $\cQ_C^\Nb$ that contains the quantum set.
\end{cor}
\begin{proof}
It follows as a corollary from Theorem \ref{thm:limitoperator} by starting from the set of words $\cS^\Nb$. 
\end{proof}

\begin{lem}
The hierarchy $\{\cQ^n_\Nb\}_{ n \in \mathbb{N}}$ cannot converge to the set of quantum assemblages.
\end{lem}
\begin{proof}
As we show in Prop.~\ref{prop:npa} below, the correlations that the Alices obtain from an assemblage in $\cQ_C^\Nb$ admit a quantum realisation in the so-called commutativity paradigm. Since these are a strict superset of the quantum correlations in the tensor-product paradigm \cite{ji2021mip}, then the claim follows. 
\end{proof}

We leave it as an open question to check how to characterise $\cQ_C^\Nb$, for instance, whether it could be expressed as a quantum-like assemblage where the Alices perform quantum actions in the commuting paradigm rather than the tensor-product one. The challenges for proving this are fundamentally the same as in the bipartite case: how to recover a tensor product structure between the Alices and Bob's spaces. We refer the reader back to Section \ref{se:convergence}. 

\section{Instrumental EPR scenarios}

The Instrumental EPR scenario is one similar to the Bob-with-Input EPR scenario, but where Bob's choice of input depends on the Alices' measurement outcomes \cite{instrumentalsteering}. Despite naively seeming like this scenario allows for communication from the Alices to Bob, the particular way in which Bob uses this information (inputting it into his device, rather that actually reading it) assures that no communication happens \cite{instrumentalsteering}. For illustration, in the bipartite case an assemblage in this scenario reads $\{ \sigma_{a|x} \}_{a\in\A, x\in\X} \equiv \As^\mathrm{I}$.  

One possibility to deal formally with assemblages in an Instrumental EPR scenario is to think of them as a post-selected Bob-with-Input assemblage. That is, to think of the assemblage elements as $\sigma_{a|x} \leftrightarrow \sigma_{a|xa}$, where $\{\sigma_{a|xy}\}_{a\in\A, x\in\X,y\in\Y}$ with $\A=\Y$ is a non-signalling assemblage in a bipartite Bob-with-Input EPR scenario. With this in mind, one can leverage the tools from Ref.~\cite{instrumentalNPA} to apply our hierarchy of tests to assemblages in Instrumental EPR scenarios. The way to test quantum-explainability of instrumental assemblages is then as follows: 

\begin{defn}\textbf{Set of assemblages in level $n$ of the instrumental EPR hierarchy: $\cQ^n_{\mathrm{I}}$}.--\\
An assemblage $\As^\mathrm{I}$ in the Instrumental EPR scenario belongs to the set $\cQ^n_{\mathrm{I}}$ iff there exists an assemblage $\As$ in a Bob-with-Input EPR scenario such that 
\begin{compactitem}
\item[(i)] $\As \in \cQ^n$, and
\item[(ii)] $\sigma_{a|x} = \sigma_{a|xa}$ with $\sigma_{a|xa} \in \As$, $\forall \, \sigma_{a|x} \in \As^\mathrm{I}$. 
\end{compactitem}
The generalisation to instrumental EPR scenarios with more than one Alice is straightforward. 
\end{defn}

\section{Our hierarchy in the context of existing work}

In this section we discuss how our hierarchy fits within and compares to existing approaches to characterising quantum behaviours `from the outside' in other scenarios related to generalised EPR scenarios.

\subsection{The Navascu\'es-Pironio-Ac\'in hierarchy for correlations in Bell scenarios}\label{npa}

Bell scenarios explore correlations between the measurement outcomes of distant parties that are only connected by a common cause (which is not necessarily in the form of a shared classical system) \cite{BellRev, Cowpie}. Bell scenarios have a resemblance to EPR scenarios, in the sense that they fundamentally explore non-classical common causes, although the information that the parties use to probe the nature of the common cause they share is different from one scenario to the other: while in Bell scenario all the parties are black-box parties, in an EPR scenario some parties (here called Bob) represent the state of their local systems by a quantum description. 

Characterising quantum correlations in Bell scenarios has been an active area of research in the last couple of decades, specially for computing the maximum advantage that quantum correlations offer as a resource to enhance our performance in communication and information processing tasks. In Bell scenarios, Navascu\'es, Pironio, and Ac\'in \cite{NPA} defined a hierararchy of semidefinite tests that converge to the set of quantum correlations\footnote{In Ref.~\cite{NPA} the set of quantum correlations is defined within the so-called commutativity paradigm, rather than the tensor-product paradigm that the reader might be familiar with. Similarly to the discussion in Sec.~\ref{se:convmul}, the NPA hierarchy does not converge to the set of tensor-product quantum correlations in Bell scenarios \cite{ji2021mip}.}. This NPA hierarchy has proven quite useful when computing upper bounds to the performance of quantum resources at some of the above-mentioned tasks \cite{BellRev}. 

One way to relate an EPR scenario with a Bell scenario is by tracing out Bob from the picture: that is, we can focus on the correlations featured by the Alices in a multipartite Bob-with-Input scenario. More precisely, given an assemblage $\As_\Nb$ in a multipartite EPR scenario, one can compute the probabilities $p(a_1 \ldots a_\Nb|x_1 \ldots x_\Nb)=\tr{\sigma_{a_1 \ldots a_\Nb|x_1 \ldots x_\Nb y}}{}$, which do not depend on the value of $y$ given the no-signalling constraints that the assemblage  $\As_\Nb$ by definition satisfies. We denote by $\Pro_\Nb$ the conditional probability distribution $\Pro_\Nb := \{p(a_1 \ldots a_\Nb|x_1 \ldots x_\Nb)\}_{a_k \in \A\,,\, x_k \in \X\,,\, k = 1 \ldots \Nb}$.

Now the question is: given an assemblage $\As_\Nb \in \cQ^n_\Nb$, what can we say about the associated correlations $\Pro_\Nb$? Below we prove that, indeed, $\Pro_\Nb$ satisfies the n-th level of the NPA hierarchy. 

\begin{prop}\label{prop:npa}
Let $\As_\Nb \in \cQ^n_\Nb$ be a multipartite assemblage in a Bob-with-Input EPR scenario. The correlations $\Pro_\Nb$ defined from it by taking $p(a_1 \ldots a_\Nb|x_1 \ldots x_\Nb)=\tr{\sigma_{a_1 \ldots a_\Nb|x_1 \ldots x_\Nb y}}{}$ belongs to the n-th level of the NPA hierarchy for the associated $\Nb$-partite Bell scenario. 
\end{prop}
\begin{proof}
The main idea of the proof is to take the certificate of order $n$ associated to $\As_\Nb$ and construct a new moment matrix from it that can act as a certificate of order $n$ associated to the correlations $\Pro_\Nb$ in the NPA hierarchy. 

First, consider a subset of the elements of $\Upsilon_\Nb$, namely $\tilde{\Upsilon}_\Nb := \Upsilon_\Nb \setminus \{ y \}_{y \in \Y}$. In addition, define a subset of the set of words $\cS_n^\Nb$ as: 
\begin{align}
\tilde{\cS}_n^\Nb := \{w \, \vert\, w \in \cS_n^\Nb \, \, \text{and $w$ is written with letters drawn from the alphabet $\tilde{\Upsilon}_\Nb$} \}\,.
\end{align}
Notice that $\tilde{\cS}_n^\Nb$ contains all and only the words that label the rows and columns of the moment matrices that act as certificates for the $n$-th level of the NPA hierarchy. 

Let $\Gamma^\Nb_n$ be a certificate of order $n$ associated to $\As_\Nb$. From it, now define the matrix $\tilde{\Gamma}^\Nb_n$ as follows: 
\begin{align}
\tilde{\Gamma}^\Nb_n(v,w) := \frac{1}{d} \tr{\Gamma^\Nb_n(v,w)}{} \quad \forall \, v,w \in \tilde{\cS}_n^\Nb \,.
\end{align}

The elements of $\tilde{\Gamma}^\Nb_n$ that will relate to the probabilities in $\Pro_\Nb$ are then: 
\begin{align}
\tilde{\Gamma}_n^{(\Nb)}(\emptyset,\mathbf{a}_\Omega | \mathbf{x}_\Omega) = p(\mathbf{a}_\Omega | \mathbf{x}_\Omega) \quad \forall \, \Omega\,,
\end{align}
with $\Omega \subset \{1,\ldots,\Nb\}$ a subset of the black-box parties, and $p(\mathbf{a}_\Omega | \mathbf{x}_\Omega)$ the marginal of $\Pro_\Nb$ given by:
\begin{align}
p(\mathbf{a}_\Omega | \mathbf{x}_\Omega) \stackrel{\text{def}}{=} \sum_{j \not\in \Omega}   \sum_{\phantom{lal} a_j \in \A} p(a_1 \ldots a_\Nb|x_1 \ldots x_\Nb)\,.
\end{align}
Notice also that $\tilde{\Gamma}_n^{(\Nb)}(\emptyset,\emptyset) =1$. 

Now, since the equivalence relations between words in $\tilde{\cS}_n^\Nb$ are the same as the ones imposed in the NPA hierarchy, the only property that we still need to prove for $\tilde{\Gamma}_n^{(\Nb)}$ to be the desired certificate is $\tilde{\Gamma}_n^{(\Nb)} \geq 0$. This follows from noticing that (i) the matrix $\Gamma^\Nb_n(v,w)$ when $v,w \in \tilde{\cS}_n^\Nb$ is a contraction of $\Gamma^\Nb_n$, and hence is positive semidefinite, and (ii) the Trace operation, together by the renormalisation by a factor of $\frac{1}{d}$ is a completely-positive operation. 
\end{proof}

\

Another way to relate an EPR scenario with a Bell scenario is by allowing Bob to make measurements on his quantum system and hence become a black-box party himself. Consider then an assemblage $\As_\Nb \in \cQ^n_\Nb$, and let Bob perform the positive operator-valued measure (POVM) on his system given by the positive semidefinite operators $\{N_k\}_{k \in \Bo}$, where $\Bo$ is the set of classical variables labelling the measurement outcomes and $\sum_{k \in \Bo} N_k = \id_{\cH_B}$. Then, we may define the following correlations:
\begin{equation*}
    p(a_1\dots a_\Nb k|x_1\dots x_\Nb y):=\frac{1}{d}\Tr{\Gamma_n^\Nb(y,\mathbf{a}_{\Omega}|\mathbf{x}_{\Omega})N_k}.
\end{equation*}
Similarly, it follows that these correlations belong to the $n$-th level of the NPA hierarchy for this $(\Nb+1)$-partite Bell scenario. If the Bob-with-Input assemblage is quantum, then so are the correlations.

\subsection{The Johnston-Mittal-Russo-Watrous hierarchy for assemblages in traditional multipartite EPR scenarios}\label{johnston}

Traditional multipartite EPR scenarios can be viewed as Bob-with-Input EPR scenarios but where the number of inputs for Bob is $1$, i.e. $\Y$ is the singleton set. There is an existing hierarchy of semidefinite programs for the traditional tripartite EPR scenario developed Johnston et al.~\cite{ENG}. We now present and generalise this hierarchy to an arbitrary number of parties. Later, we show how it emerges as a special case of the hierarchy developed herein. 

Recall the traditional multipartite EPR scenario that consists of $\Nb+1$ parties, of which $\Nb$ are black-box parties with input and output variables taking values within $\X$ and $\A$ respectively \cite{pqs}. Here, the assemblage elements read $\sigma_{a_1 \ldots a_\Nb | x_1 \ldots x_\Nb}$. The remaining party (Bob) holds a quantum system associated with a Hilbert space of dimension $d$. Thus the alphabet $\Upsilon^{\mathrm{J}}_{\Nb}$ that we  use to define the words (from inputs and outputs) is composed of the following elements: 
\begin{align}
\Upsilon^{\mathrm{J}}_{\Nb} := \{\emptyset\} \cup \{ a_1|x_1 \}_{{ x_1 \in \X\,, a_1 \in \A \setminus \{0\}}} \cup \ldots \cup \{ a_\Nb |x_\Nb \}_{{ x_\Nb \in \X\,, a_\Nb \in \A \setminus \{0\}}} \,. 
\end{align}

The definition of words, their concatenation, and their ``dagger'' follows similarly from Sec.~\ref{se:BwIhier}. We denote by $\tilde{\cS}^*_\Nb$ the set of all words of arbitrary length with letters drawn from $\Upsilon^{\mathrm{J}}_\Nb$.

The symmetry operations on the elements of the alphabet generalise to the following:
\begin{compactitem} 
\item $\mathbf{v}\mathbf{w} \equiv \mathbf{v} \emptyset \mathbf{w} $ for all $\mathbf{u},\mathbf{w} \in \tilde{\cS}^*_\Nb$,
\item $\mathbf{v}\mathbf{v} \equiv \mathbf{v}$ for all $\mathbf{v} \in \Upsilon^{\mathrm{J}}_\Nb$,
\item $a_k|x_k a_{k'}|x_{k'}  \equiv a_{k'}|x_{k'} a_k|x_k$ for all $a_k|x_k,a_{k'}|x_{k'} \in \Upsilon^{\mathrm{J}}_\Nb$ and $k\neq k' \in \{1,\ldots,\Nb\}$.
\end{compactitem} 
Definitions of equivalent and null words follow similarly from Sec.~\ref{multibwi}. As before, with slight abuse of notation, we can specify the ``set of words of certain length'' and ``moment matrix of order $n$'':
\begin{defn} \textbf{The set of words} $\boldsymbol{\tilde{\cS}_n^\Nb}$.--\\
A word $\mathbf{v} \in \tilde{\cS}^*_\Nb$ belongs to $\tilde{\cS}_n^\Nb$ if it may arise from concatenating at most $n$ letters of the alphabet $\Upsilon^{\mathrm{J}}_\Nb$.
\end{defn}

\begin{defn}\label{def:MMnMENG}\textbf{Moment matrix of order $\mathbf{n}$ with $\Nb$ Alices: $\boldsymbol{\Delta_n^{(\Nb)}}$}.--\\
Let $\Delta_n^{(\Nb)}$ be a matrix of size $\lvert \tilde{\cS}_n^\Nb \lvert \times \lvert \tilde{\cS}_n^\Nb \lvert$, whose entries are $d \times d$ complex matrices, and whose rows and columns are indexed by the words in $\tilde{\cS}_n^\Nb$. This matrix is a moment matrix of order $n$ iff it satisfies the following properties: 
\begin{align}
\Delta_n^{(\Nb)} \geq 0 \,, \label{eq:2agg}\\
\Delta_n^{(\Nb)}(\mathbf{v},\mathbf{w}) &= \Delta^\Nb_n(\mathbf{v'},\mathbf{w'}) \quad \text{if} \quad \mathbf{v}^\dagger\mathbf{w} \equiv \mathbf{v'}^\dagger\mathbf{w'} \,,\label{eq:4agg}\\
\Delta_n^{(\Nb)}(\mathbf{v},\mathbf{w}) &= \mathbf{0}_d  \quad \text{if} \quad \mathbf{v}^\dagger\mathbf{w}  \, \text{is null} \,,\label{eq:5agg}
\end{align}
where $\mathbf{0}_d$ is the $d \times d$ matrix whose all entries are 0. 
\end{defn}

Finally, the definitions of ``Certificate of order $n$ for an assemblage'' and ``Set of assemblages in level $n$ of the hierarchy'' are as follows: 

\begin{defn}\label{def:certNbENG} \textbf{Certificate of order $n$ for an assemblage $\As_\Nb$}.--\\
Let $\As_\Nb = \{\sigma_{a_1 \ldots a_\Nb|x_1 \ldots x_\Nb }\}_{a_k \in \A, \, x_k \in \X, \, k \in \{1,\ldots,\Nb\}}$ be a non-signalling assemblage in the multipartite traditional EPR scenario with $\Nb$ black-box parties \cite{pqs}.

Let $\Omega \subset \{1,\ldots,\Nb\}$ be a subset of the black-box parties, and denote by $\sigma_{\mathbf{a}_\Omega | \mathbf{x}_\Omega}$ the assemblage elements that result when tracing out the parties not in $\Omega$, i.e.,
\begin{align}\label{eq:32a}
\sigma_{\mathbf{a}_\Omega | \mathbf{x}_\Omega} \stackrel{\text{def}}{=} \sum_{j \not\in \Omega}   \sum_{\phantom{a} a_j\in \A} \sigma_{a_1 \ldots a_\Nb|x_1 \ldots x_\Nb}\,.
\end{align}

A matrix $\Delta_n^{(\Nb)}$ is a certificate of order $n$ for $\As_\Nb$ iff $\Delta_n^{(\Nb)}$ is a moment matrix of order $n$ for $\Nb$ Alices, and the following are satisfied: 
\begin{align}
\Delta_n^{(\Nb)}(\emptyset,\emptyset) &= \, \rho_{B} \,,\label{eq:103ag}\\
\Delta_n^{(\Nb)}(\emptyset, \mathbf{a}_\Omega | \mathbf{x}_\Omega) &= \, \sigma_{\mathbf{a}_\Omega | \mathbf{x}_\Omega} \quad \forall a_k|x_k \in \Upsilon^{\mathrm{J}}_\Nb\,, \quad k \in \Omega\,,\quad \forall \Omega\,,\label{eq:104ag}
\end{align}
where $\rho_{B} = \sum_{a_1\in\A\,,\ldots \,, a_\Nb\in\A}  \sigma_{a_1 \ldots a_\Nb|x_1 \ldots x_\Nb}$ is the marginal state of Bob's system.
\end{defn}

\begin{defn}\label{def:engours}\textbf{Set of assemblages in level $n$ of the Johnston-Mittal-Russo-Watrous multipartite hierarchy \cite{ENG}: $\tilde{\cQ}^n_\Nb$}.--\\
An assemblage $\As_\Nb$ in the multipartite EPR scenario with $\Nb$ black-box parties belongs to the set $\tilde{\cQ}^n_\Nb$ iff there exists a matrix $\Delta_n^{(\Nb)}$ that is a certificate of order $n$ for $\As_\Nb$. 
\end{defn}

Def.~\ref{def:engours} presents the Johnston-Mittal-Russo-Watrous hierarchy in a notation tailored at this paper. To see how this corresponds to the original definition in Ref.~\cite{ENG}, notice that their matrix $M^{(\Nb)}$ has elements define as in Eq.~(32) of Ref.~\cite{ENG} with $M^{(\Nb)}(\mathbf{\alpha},\mathbf{\beta}) \equiv \Delta_n^{(\Nb)}(\mathbf{\alpha},\mathbf{\beta})$ for all $\mathbf{\alpha},\mathbf{\beta} \in \tilde{\cS}_n^\Nb$, since the set of words in both works is defined analogously. Finally, the various conditions on the matrix $M^{(\Nb)}$ are defined analogously to those for the matrix $\Delta_n^{(\Nb)}$ (e.g., our null condition in Eq.~\eqref{eq:5agg} is equivalent to that in Eq.~(25) of Ref.~\cite{ENG}).

\bigskip

Now we show that the Johnston-Mittal-Russo-Watrous hierarchy may arise as a special case of the multipartite Bob-with-Input hierarchy introduced herein. 

\begin{thm}\textbf{Johnston-Mittal-Russo-Watrous hierarchy from the Bob-with-Input one}.--\\
Let $\As_\Nb$ be an assemblage in a multipartite EPR scenario. Let $\Y$ be a singleton set and define the assemblage $\As_{\Nb}^*$ as that with elements $\sigma^*_{a_1 \ldots a_\Nb | x_1 \ldots x_\Nb y} := \sigma_{a_1 \ldots a_\Nb | x_1 \ldots x_\Nb}$. Then, 
\begin{align*}
\As_\Nb \in \tilde{\cQ}^n_\Nb \quad \Rightarrow \quad \As_{\Nb}^* \in \cQ^n_\Nb\,.
\end{align*}
Conversely, let $\As_{\Nb}^*$ be an assemblage in a multipartite Bob-with-Input scenario where $\Y$ is the singleton set. Define the assemblage $\As_\Nb$ in a multipartite EPR scenario as that with elements $\sigma_{a_1 \ldots a_\Nb | x_1 \ldots x_\Nb} := \sigma^*_{a_1 \ldots a_\Nb | x_1 \ldots x_\Nb y}$. Then, 
\begin{align*}
\As_{\Nb}^* \in \cQ^n_\Nb \quad \Rightarrow \quad \As_\Nb \in \tilde{\cQ}^n_\Nb\,.
\end{align*}
\end{thm}

\begin{proof}
Let's start from the second statement. Consider an assemblage $\As^*_\Nb$ that belongs to the set ${\cQ}^n_\Nb$. By definition, there exists a moment matrix  $\Gamma_n^{(\Nb)}$ that is a certificate of order $n$ for $\As^*_\Nb$. From $\Gamma_n^{(\Nb)}$, we can now construct a matrix $\Delta_n^{(\Nb)}$ with rows and columns labelled by words in $\tilde{\cS}_n^\Nb$, as follows:
\begin{align}
    \Delta_n^{(\Nb)}(\mathbf{v},\mathbf{w}) := d\,\Gamma_n^{(\Nb)}(\mathbf{v},\mathbf{w}y)^\mathrm{T} \quad \forall \, \mathbf{v},\mathbf{w} \in \tilde{\cS}_n^\Nb\,.
\end{align}
By definition, since $\Gamma_n^{(\Nb)}$ is a moment matrix of order $n$ for the multipartite Bob-with-Input scenario, then  $\Delta_n^{(\Nb)}$ is a moment matrix of order $n$ for the multipartite EPR scenario -- one needs only to compare how the constraints of Def.~\ref{def:MMnM} imply that those in Def.~\ref{def:MMnMENG} are satisfied\footnote{That $\Gamma_n^{(\Nb)} \geq 0$ follows from the fact that it arises by applying a positive map to a contraction of $\Delta_n^{(\Nb)}$, where the latter is itself positive semidefinite. }. In addition, the definition of the elements of $\As_\Nb$ from those of $\As^*_\Nb$, together with the definition of the elements $\Delta_n^{(\Nb)}(\emptyset, \mathbf{a}_\Omega | \mathbf{x}_\Omega)$, imply that $\Delta_n^{(\Nb)}$ is a certificate of order $n$ for $\As_\Nb$. Therefore, $\As_\Nb \in \tilde{\cQ}^n_\Nb$. 

Now let us prove the converse direction. We start from an assemblage $\As_\Nb$ that belongs to the set $\tilde{\cQ}^n_\Nb$. Let $\Delta_n^{(\Nb)}$ be a moment matrix associated to $\As_\Nb$. Then one can define a matrix $\Gamma_n^{(\Nb)}$ whose rows and columns are labelled by the words in $\cS_n^\Nb$ as follows: 
\begin{align}
  \Gamma_n^{(\Nb)}(\mathbf{\alpha},\mathbf{\beta}) &:= \Tr{\Delta_n^{(\Nb)}(\mathbf{\alpha},\mathbf{\beta})} \, \id_d \quad \forall \, \mathbf{\alpha},\mathbf{\beta} \in \tilde{\cS}_n^\Nb\,,\\
  \Gamma_n^{(\Nb)}(\mathbf{v},\mathbf{w}y) &:= \frac{1}{d}  \Delta_n^{(\Nb)}(\mathbf{\alpha}',\mathbf{\beta}')^\mathrm{T} \quad \forall \, \mathbf{v},\mathbf{w} \in {\cS}_n^\Nb\,,
\end{align}
where $\mathbf{\alpha}' \in \tilde{\cS}_n^\Nb$ is the word that arises from $\mathbf{v}$ by removing from it the letter $y \in \Y$ (should it appear), and similarly for $\mathbf{\beta}'$.

Notice that, by construction, $\Gamma_n^{(\Nb)}$ satisfies the constraints in Eqs.~\eqref{eq:3ag} to \eqref{eq:6ag} and Eqs.~\eqref{eq:7ag} to \eqref{eq:9ag}, since $\Delta_n^{(\Nb)}$ is a valid moment matrix for the assemblage $\As_\Nb$. What remains to be shown is that $\Gamma_n^{(\Nb)} \geq 0$.

To see that $\Gamma_n^{(\Nb)} \geq 0$, let us first construct the following matrix:
\begin{align}
    \gamma_n^{(\Nb)} = \begin{bmatrix}
    \mathcal{A} & \mathcal{A} \\
    \mathcal{A} & \mathcal{B}
    \end{bmatrix}\,,
\end{align}
with 
\begin{align}
    \mathcal{A}(\alpha,\beta) &:= \frac{1}{d} \,  \Delta_n^{(\Nb)}(\alpha,\beta)^\mathrm{T}\,, \quad \forall \, \mathbf{\alpha},\mathbf{\beta} \in \tilde{\cS}_n^\Nb\,, \\
    \mathcal{B}(\alpha,\beta) &:= d \, \Tr{\mathcal{A}(\alpha,\beta)}\, \id_d\,, \quad \forall \, \mathbf{\alpha},\mathbf{\beta} \in \tilde{\cS}_n^\Nb \,.
\end{align}
Now define the completely-positive map $\xi$ as:
\begin{align*}
    \xi[(\cdot)] := d\,\Tr{(\cdot)} \, \id_d - (\cdot)\,.
\end{align*}
One can then see that 
\begin{align*}
    \mathcal{B} - \mathcal{A} = (\id_{|\tilde{\cS}_n^\Nb|} \otimes \xi) [\mathcal{A}] \geq 0 \quad \text{since} \quad \mathcal{A} \geq 0\,.
\end{align*}
Therefore, $\gamma_n^{(\Nb)} \geq 0$. Now one only needs to see that $\gamma_n^{(\Nb)} \geq 0 \,\, \Rightarrow \,\, \Gamma_n^{(\Nb)} \geq 0$. For this, note that 
\begin{align}
    \mathcal{B}(\alpha,\beta) &\equiv \Gamma_n^{(\Nb)}(\alpha,\beta)\,, \quad \forall \, \mathbf{\alpha},\mathbf{\beta} \in \tilde{\cS}_n^\Nb\,,\\
    \mathcal{A}(\alpha,\beta) &\equiv \Gamma_n^{(\Nb)}(\alpha,\beta y)\,, \quad \forall \, \mathbf{\alpha},\mathbf{\beta} \in \tilde{\cS}_n^\Nb\,,
\end{align}
and in general
\begin{align}
    \Gamma_n^{(\Nb)}(\mathbf{v},\mathbf{w}y)  \equiv \mathcal{A}(\alpha',\beta')\,, \quad \forall \, \mathbf{v},\mathbf{w} \in {\cS}_n^\Nb\,,
\end{align}
where, as before, $\mathbf{\alpha}' \in \tilde{\cS}_n^\Nb$ is the word that arises from $\mathbf{v}$ by removing from it the letter $y \in \Y$ (should it appear), and similarly for $\mathbf{\beta}'$. This means that $\Gamma_n^{(\Nb)}$ can be constructed from $\gamma_n^{(\Nb)}$ by duplicating the rows and columns of $\gamma_n^{(\Nb)}$ (those pertaining to the blocks $[\mathcal{A} \, \mathcal{A}]$ and $[\mathcal{A} \, \mathcal{A}]^\mathrm{T}$, respectively) as many times as follows. Therefore, $\gamma_n^{(\Nb)} \geq 0 \,\, \Rightarrow \,\, \Gamma_n^{(\Nb)} \geq 0$.

\end{proof}

\subsection{Our hierarchy applied to traditional bipartite EPR scenarios}

Applying our hierarchy to a traditional bipartite EPR scenario means applying it to a scenario with only one Alice, and where Bob has only one possible input -- $\Y$ is the singleton set. In this scenario, the GHJW theorem \cite{gisin1989stochastic,hughston1993complete} shows that any non-signalling assemblage has a quantum realisation. Hence, one expects that in this case our hierarchy collapses, and that the assemblages that belong to  $\cQ^1$ already admit a quantum realisation. The proof of this claim is presented in Theorem 9 of Ref.~\cite{rossi2022characterising}. 

\section{Conclusions and outlook}

In this paper we developed a hierarchy of semidefinite tests to assess whether an assemblage in Bob-with-Input EPR scenarios (and Instrumental EPR scenarios therefore) may not admit of a quantum explanation. This hierarchy allows one not only to test whether an assemblage cannot arise within a quantum EPR experiment, but also serve as a tool to compute upper bounds on the violations of steering inequalities, as well as on the  performance of quantum assemblages in certain communication and information-processing tasks (including generalisations of extended nonlocal games). Our hierarchy has already been used to certify post-quantumness of some specific assemblages \cite{bpqs}. 

Looking into the future, one technical challenge is to give a natural characterisation of the set of assemblages to which the hierarchy converges. We know that the set to which it converges necessarily contains the set of quantum assemblages, but it could be larger still. Along these lines, we have shown that the NPA hierarchy can be retrieved within our hierarchy in Section \ref{npa}. It was demonstrated in Ref. \cite{Navascu_s_2008} that the NPA hierarchy converges to the set of correlations where parties making commuting measurements on a shared quantum state. What is the analogue to this result for all scenarios captured by our hierarchy? 


On the foundational side, one interesting open question is whether our hierarchy helps to find physical principles from which to derive the set of quantum assemblages (i.e. ruling out the post-quantum ones). This could be pursued by searching for a hierarchy of physical principles whose levels are related to the levels of our hierarchy one-to-one. If such a connection is found, this will have immense implications not only in the task of ``characterising quantum assemblages from physical principles'' but also in the task of ``characterising quantum correlations in Bell scenarios from physical principles'' -- an open question on which the community has done intense work over more than a decade -- given the connections between the NPA hierarchy and ours.

\section*{Acknowledgments}

We thank Tobias Fritz and Tim Netzer for introducing ABS to TD. 
MJH and ABS acknowledge the FQXi large grant ``The Emergence of Agents from Causal Order'' (FQXi FFF Grant number FQXi-RFP-1803B).
ABS acknowledges support by the Foundation for Polish Science (IRAP project, ICTQT, contract no. 2018/MAB/5, co-financed by EU within Smart Growth Operational Programme). 
All of the diagrams within this manuscript were prepared using TikZit.

\bibliographystyle{quantum}
\bibliography{bibliography.bib}

\end{document}